\documentclass{imamci}

\jno{dnnxxx}
\usepackage{natbib}
\usepackage{graphicx}
\usepackage{amsmath,amsthm,bm,mathrsfs}
\usepackage{comment} 
\usepackage{amssymb}
\renewcommand{\d}{{\mathrm{d}}}
\newcommand{\m}{{\mathrm{m}}}
\newcommand{\hu}{\hat{u}}
\newcommand{\hy}{\hat{y}}
\newcommand{\tu}{\tilde{u}}
\newcommand{\ty}{\tilde{y}}

\newcommand{\IN}{I^{(N)}}
\newcommand{\Il}{I^{(L)}}
\usepackage{color}
\usepackage{subfigure}
\usepackage{algorithmic}
\usepackage{listings}
\newtheorem{prop}{Proposition}
\renewenvironment{proof}[1][Proof]{\noindent\textit{#1. } }{\hfill$\square$}

 \newtheoremstyle{theorem}{6pt}{6pt}{\rm}{}{\sffamily}{ }{ }{}
 \theoremstyle{theorem}
\newtheorem{theorem}{\sc Theorem}[section]

 \newtheoremstyle{algorithm}{6pt}{6pt}{\rm}{}{\sffamily}{ }{ }{}
 \theoremstyle{algorithm}
\newtheorem{algorithm}{\sc Algorithm}[section]

 \newtheoremstyle{lemma}{6pt}{6pt}{\rm}{}{\sffamily}{ }{ }{}
 \theoremstyle{lemma}
 \newtheorem{lemma}{\sc Lemma}[section]

\newtheoremstyle{case}{6pt}{6pt}{\rm}{}{\sffamily}{. }{ }{}
 \theoremstyle{case}
\newtheorem{case}{{\sc Case}}

 \newtheoremstyle{statement}{6pt}{6pt}{\rm}{}{\sffamily}{ }{ }{}
\theoremstyle{statement}

 \newtheoremstyle{corollary}{6pt}{6pt}{\rm}{}{\sffamily}{ }{ }{}
 \theoremstyle{corollary}
 \newtheorem{corollary}{\sc Corollary}[section]
 
  \newtheoremstyle{definition}{6pt}{6pt}{\rm}{}{\sffamily}{ }{ }{}
 \theoremstyle{definition}
 \newtheorem{definition}{\sc Definition}[section]

\newtheoremstyle{example}{6pt}{6pt}{\rm}{}{\sffamily}{ }{ }{}
\theoremstyle{example}
\newtheorem{example}[theorem]{\sc Example}

\newtheoremstyle{remark}{6pt}{6pt}{\rm}{}{\sffamily}{ }{ }{}
\theoremstyle{remark}
\newtheorem{remark}{\sc Remark}[section]

\newtheoremstyle{approximation}{6pt}{6pt}{\rm}{}{\sffamily}{ }{ }{}
\theoremstyle{approximation}

\newtheoremstyle{scheme}{6pt}{6pt}{\rm}{}{\sffamily}{ }{ }{}
\theoremstyle{scheme}

\newtheoremstyle{Algorithm}{6pt}{6pt}{\rm}{}{\sffamily}{ }{ }{}
\theoremstyle{Algorithm}

\newtheoremstyle{Assumption}{6pt}{6pt}{\rm}{}{\sffamily}{ }{ }{}
\theoremstyle{Assumption}

\newtheoremstyle{proposition}{6pt}{6pt}{\rm}{}{\sffamily}{ }{ }{}
\theoremstyle{proposition}

\newtheoremstyle{hypo}{6pt}{6pt}{\rm}{}{\sffamily}{ }{ }{}
 \theoremstyle{hypo}

  \newtheoremstyle{Step}{6pt}{6pt}{\rm}{}{}{ }{ }{}
 \theoremstyle{Step}

\numberwithin{equation}{section}

\begin{document}

%%%%%%%%%%%%%%%%%
\title{Method for estimating hidden structures determined by unidentifiable state-space models and time-series data based on the Gr\"{o}bner basis}
\author{ {\sc Mizuka Komatsu and Takaharu Yaguchi}\\[2pt]
Graduate School of System Informatics, Kobe University,  \\
 1-1 Rokkodai-cho, Nada-ku, Kobe 657-8501, Japan.\\[6pt]
%{\rm [Received on 18 February 2005]}\vspace*{6pt}
}
\pagestyle{headings}
%\markboth{C. E. POWELL}{\rm PARAMETER-FREE H(DIV) PRECONDITIONING}
\markboth{Komatsu, M. \& Yaguchi, T.}%{\rm Method for estimating hidden structures determined by unidentifiable state-space models and time-series data based on the Gr\"obner basis}
{\rm %METHOD FOR 
ESTIMATING HIDDEN STRUCTURES 
DETERMINED 
BY %UNIDENTIFIABLE 
MODELS \& %TIME-SERIES 
DATA 
BASED ON %THE 
GR\"OBNER BASIS}
\maketitle

%%%%%%%%%%%%%%%%%abstract style
%Two grouping braces are necessary in abstract environment
%first argument contains abstract text; second argument contains keywords
%text

\begin{abstract}
{In this study, we propose a method for extracting the hidden algebraic structures of model parameters that are uniquely determined by observed time-series data and unidentifiable state-space models, explicitly and exhaustively.
	State-space models are often constructed based on the domain, for example, physical or biological. Such models include parameters that are assigned specific meanings in relation to the system under consideration, 
	which is examined by estimating the parameters using the corresponding data.
	As the parameters of unidentifiable models cannot be uniquely determined from the given data, it is difficult to examine the systems described by such models. 
	To overcome this difficulty, multiple possible sets of parameters are estimated and analysed in the exiting approaches; however, in general, all the possible parameters cannot be explored; therefore, considerations on the system using the estimated parameters become insufficient.

	In this study, focusing on certain structures determined by the observed data and models uniquely, even if they are unidentifiable, we introduce the concept of parameter variety. This is newly defined and proven to form algebraic varieties, in general. A computational algebraic method that relies on the Gr\"{o}bner basis for deriving the explicit representation of the varieties is presented along with the supporting theory.
Furthermore, its application in the analysis of a model that describes virus dynamics is presented.
With this, new insight on the dynamics overlooked by the conventional approach are discovered, confirming the applicability of our idea and the proposed method.}
{unidentifiable model, state-space model, parameter estimation, Gr\"{o}bner basis, algebraic variety.}
\end{abstract}
%%%%%%%%%%%%%%%%%%%%%%%%%%%%%%%

\section{Introduction}
In this study, we deal with the problem of analysing systems based on the observed time-series data using state-space models that describe the systems. Along with the supporting theory, we propose a method for extracting all the feasible parameters of unidentifiable state-space models, given time-series data,
explicitly, based on our novel concept defined as the parameter variety.

In various fields such as systems biology and control, systems under considerations are modelled as and investigated through mathematical models \cite{nonlinearsys, science, pbpk}. In particular, if the models have unknown parameters, they are estimated using observed time-series data, etc. An illustrative example of such a parameter identification problem is presented below.

Suppose we have a state-space model that describes the competition within two species:
\begin{align}
  \frac{\d x_1}{\d t} &= a_1x_1\left(1-\frac{x_1}{a_2}-a_3\frac{x_2}{a_2}\right), \label{eq:lv1} \\
  \frac{\d x_2}{\d t} &= a_4x_2\left(1-\frac{x_2}{a_5}-a_6\frac{x_1}{a_5}\right), \label{eq:lv2}\\
    y &= x_1, \label{eq:lvout}
\end{align}
where $x \in \mathbb{R}^2$ is the state variable vector that denotes the number of the species; $y(t) \in \mathbb{R}$ is the output to be observed; and $a = {\left(a_1, \ldots, a_6 \right)}^\top\in \mathbb{R}^6$ is the unknown parameter vector, whose elements are positive constants. The set of ordinary differential equations \eqref{eq:lv1}--\eqref{eq:lv2} belongs \textcolor{black}{to the well-known Lotka-Volterra competition model \cite{Murray}}. Typically, when modelling a system under consideration, a specific meaning in relation to the modelled system is assigned to each parameter: $a_1, a_4$ birth rates, $a_2, a_5$ carrying capacities, and $a_3, a_6$ effect on the other species. The set of parameters $a$ is estimated using observed time-series data, where $x_2$ is not observed but $x_1$ is, as shown in \eqref{eq:lvout}. For simplicity, we assume that an artificial time-series data generated by \eqref{eq:lv1}, \eqref{eq:lv2}, \eqref{eq:lvout} given $a = \left(1, 0.5, 5, 1, 0.2, 2.4\right) =: a^{*}$, $x(0) = (1,2)^T$, is observed as depicted in %\ref{fig:illust_graph}
Fig. \ref{illust_graph}. It would be satisfactory, if $a^{*}$ is correctly estimated for fitting the model output $y$ to the data using methods such as the Newton method; thereby, the system, i.e., the competitions within the two species, can be investigated using the estimated parameters.
 For instance, the effect on the first species by the second is revealed to be greater than vice versa because $a_3^* < a_6^*$.
 \begin{figure}
  \centering \label{illustration}
  \subfigure[Observed time-series data (black dots) and the output of model \eqref{eq:lv1}, \eqref{eq:lv2}, \eqref{eq:lvout} given $a^* = (1, 0.5, 5, 1, 0.2, 2.4)$ and $a^{**} =(1, 0.5, 0.9, 1, 1.1, 12.8)$ (red solid line and blue dashed line, respectively).]{%
            \includegraphics[clip, width=0.45\textwidth]{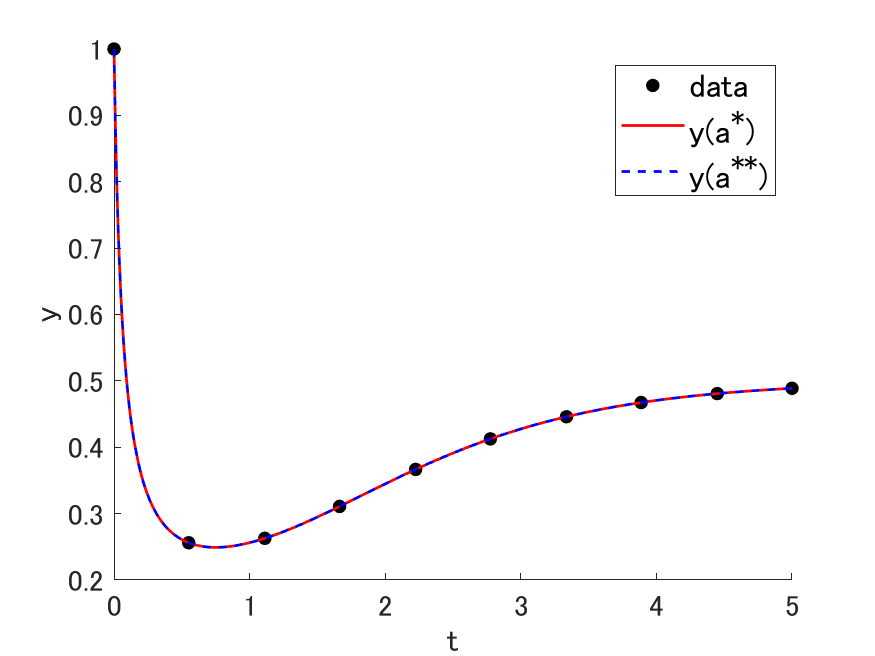}
         \label{illust_graph} }
  %\hfill
            \subfigure[Estimated parameters depicted in three-dimensional parameter space. $a^{*}$ and $a^{**}$ are denoted by the red circle and blue square, respectively. $a^{*}, a^{**}$, and the parameters denoted by black circles generate the same model output $y$ that fits the observed data.]{%
            \includegraphics[clip, width=0.45\textwidth]{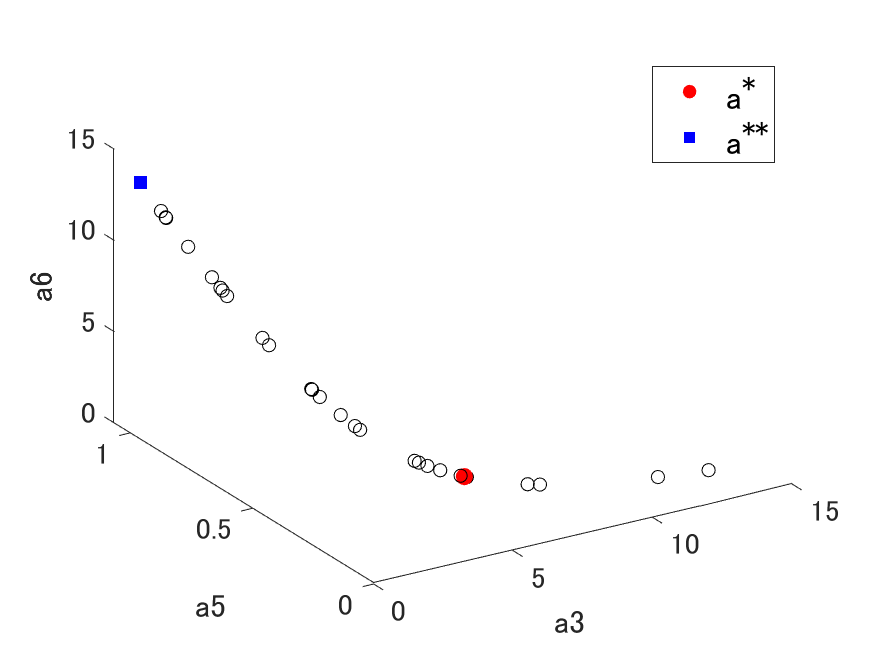}
     \label{illust_pv}
     }
             \caption{Illustrative example of an unidentifiable state-space model and its parameter estimation}
  \end{figure}

However, if the estimation method is changed, e.g., given different \textcolor{black}{initial parameter values for the Newton iteration},
different parameters can be estimated for \eqref{eq:lv1}, \eqref{eq:lv2}, \eqref{eq:lvout}. As shown in Fig. %%\ref{fig:illust_graph}
\ref{illust_graph}, estimated parameter $a^{**} = \left(1, 0.5, 0.9, 1, 1.1, 12.8\right)$ generates the same output trajectory as the one generated given $a^*$. Moreover, as $a_3^{**} > a_6^{**}$, the opposite interpretation for the system compared to the one obtained with $a^*$ is acquired. In addition, a number of parameters exist with which the output fits the data, as shown in Fig. %\ref{fig:illust_pv}
\ref{illust_pv}. Hence, the interpretation of the estimated results is revealed to be no longer straightforward. Such a situation is not only problematic, but can easily occur in real applications \cite{meshkatn, science}, which are at times overlooked as mentioned in \cite{reverse_eng_biol}.

In this study, under certain assumptions, we consider the above-mentioned state-space models, called ``unidentifiable" models. Precisely defined later in the paper, an ``unidentifiable" model can be considered as a type of model, 
whose parameters cannot be uniquely determined from the given data \cite{ljung1994, miao, meshkatl, meshkatn}. Mainly, there are two methods to deal with this type of model:

\begin{itemize}
  \item Rearrange the model into an identifiable one such that there exists a unique parameter $a$ for $y$ to fit the observed time-series data.
  \item Estimate multiple parameters such that $y$ fits the observed data, and analyse them to examine the parameter relationships. %knowledge on parameters of interest.
\end{itemize}
If modifications in the model are allowed and are reasonable, the first approach, which has been well-studied \cite{miao, meshkatl, zvi,chappell, evans}, is valid. %Thereby the investigation of the system becomes straightforward.
In this case, the system under consideration can be investigated using the estimated parameters. However, if this is not the case, the second approach needs to be used.

We focus on the latter case and propose an alternative approach in this study.
In some applications, it is not possible or reasonable to rearrange the model because of constraints such as those on experiments \cite{pbpk}. Thus, multiple parameters are often estimated and analysed \cite{Hengl_2007, profilelikelihood1, profilelikelihood2}.
In \cite{Hengl_2007}, a method to estimate the functional relations between parameters using the estimated multiple parameters is introduced. This approach implicitly assumes that it is impossible to obtain all the possible parameters, and instead, the relationships are investigated in a data-driven manner. Compared to this, the profile likelihood method \cite{profilelikelihood1, profilelikelihood2} enumerates the feasible parameters at least locally in a numerical sense; 
the parameters are first estimated by assigning a specific value to one of the parameters. The specific value then is shifted slightly. These two steps are alternately repeated for each parameter. However, the parameter relationships that can be estimated with this method are restricted to a case where they form one-dimensional manifolds, due to its algorithmic nature. In addition, although these existing methods provide practical frameworks for investigating the parameters of unidentifiable models with the observed data, theoretical studies on such parameters have been limited, to the best of our knowledge.

In view of the above, in this study, the algebraic geometric structures of the feasible parameters of unidentifiable models, given observed data, are investigated theoretically, restricting the models to ordinary differential equations with polynomial terms of the state and input variables. 
More precisely, under certain additional assumptions, 
\begin{enumerate}
  \item we show that the possible parameters determined by the given unidentifiable state-space model and the observed time-series data generally form an algebraic variety by applying 
commutative algebra and algebraic geometry
and
  \item we propose a method to extract the variety, i.e., enumerate all the possible parameters by combining computational algebraic techniques and numerical calculations.
\end{enumerate}

Considering that the state-space models and time-series data determine certain geometric structures, the key idea of our approach is to extract all the possible parameters, even if the models are unidentifiable. We define our novel concept as the parameter variety, in this study. Such structures may be implicitly investigated through existing approaches such as \cite{Hengl_2007, profilelikelihood1, profilelikelihood2}.

In our approach,
we first derive them mathematically, and then estimate them given the temporal measurements. Once the structure, i.e., the variety is extracted, overlooking of the feasible parameters, which may lead to insufficient or inappropriate system considerations, would never occur. In addition, as proof of the derivation of the parameter variety, we demonstrate the application of the proposed approach in the analysis of viral dynamics, which reveals an important fact on the efficacy of the drug that was missed
in a previous study \cite{science}.

\begin{remark}\label{remark:differentialalgebra}
Those familiar with differential algebra \cite{ritt}, particularly with structural identifiability analysis through differential algebraic approaches such as \cite{ljung1994} and \cite{meshkatn}, may know that the feasible parameters of unidentifiable models may form algebraic varieties. However, we demonstrate this formally, through its applicability in the modelling approach mentioned above. Besides, technically, this study is based on nondifferential algebra and algebraic geometry, considering that it may be preferable to not differentiate the indeterminants more than necessary, in view of the theory of differential equations. The following truncated differential ideal can be regarded as a non-differential one considering each of the derivatives as indeterminants:
\begin{align*}
&[\left\{ p\right\}] \cap K[x_1,\ldots, x_N, x_1^{(1)}, \ldots, x_N^{(1)}, \ldots, x_1^{(s)}, \ldots, x_N^{(s)}] \\
= &\langle \left\{ p\right\}, \{ p^{(1)}\}, \ldots, \{ p^{(s)}\} \rangle \subset K[x_1,\ldots, x_N, x_1^{(1)}, \ldots, x_N^{(1)}, \ldots, x_1^{(s)}, \ldots, x_N^{(s)}],
\end{align*}
where $\left\{ p\right\}$ is a finite set of differential polynomials for which the indeterminants are $x_1, \ldots, x_N$, and their derivatives with respect to time over the differential polynomial ring over field $K$ are $K\left\{x_1,\ldots, x_N \right\}$. $[\left\{ p\right\}]$ and $\langle \left\{ p\right\} \rangle $ are the differential and non-differential ideals generated by $\left\{ p\right\}$, respectively. The superscript numbers in brackets denote the order of the derivative. In terms of the algebraic geometry, the algebraic variety defined by ideal $[\left\{ p\right\}] \cap K[x_1,\ldots, x_N, x_1^{(1)}, \ldots, x_N^{(1)}, \ldots, x_1^{(L)}, \ldots, x_N^{(L)}]$, where $l \geq s$ is a subset of the variety defined by $\langle \left\{ p\right\}, \{ p^{(1)}\}, \ldots, \{ p^{(s)}\} \rangle$, clearly shows that the algebraic approach works for relaxing the infinite differentiation property of differential algebra.
\end{remark}

The remainder of the paper is organized as follows. In %Section \ref{sec:definition}
Section \ref{sec:definition}, the definitions related to the parameter variety are provided. In %Section \ref{sec:method}
Section \ref{sec:method}, 
we present a method that extracts an algebraic variety containing the parameter variety of an unidentifiable model given the observed data. This section is organized such that users can focus on the specific usage of the method without considering the technical details of the algebra. Further, the application of the proposed method in viral dynamic analysis is described in %Section \ref{sec:app}
Section \ref{sec:app}. In this section, in particular, our idea and its applicability are clarified through real world application.
In %Section \ref{sec:theory}
Section \ref{sec:theory}, the algebraic preliminaries and mathematical formulations of our problem are stated, followed by main theorems. In short, the parameter variety of a state-space model, given the observed time-series data, generally forms an algebraic variety.

\section{Definitions}
\label{sec:definition}
In this section, after describing the models and data to be considered, we
define the parameter variety. % is defined.
In this study, we deal with the following state-space models:
\begin{align}
\frac{\d x}{\d t} &= f(x,u;a), \label{eq:math_eq}\\
y &= g(x,u;a), \label{eq:out_eq}
\end{align}
%given partially known initial conditions of $y=y(T_0) =:y_0$,
where $x(t) \in \mathbb{R}^N$ is the state variable vector, $u(t) \in
\mathbb{R}^M$ is the input vector, $y(t) \in \mathbb{R}$ is the output,
and $a \in \mathbb{R}^n$ is the unknown parameter vector, where $N, M, n$ are positive integers. $u(t)$ and $y(t)$ are observed and denoted by
$\hu(t)$ and $\hy(t)$, when it is assumed that they are given as data. We
suppose that $u(t)$ is not constrained by \eqref{eq:math_eq},
\eqref{eq:out_eq}, and $x(t),u(t),y
(t)$ are sufficiently smooth as mentioned later in
%\ref{remark:smoothness}
Remark \ref{remark:smoothness}. \eqref{eq:math_eq} is the mathematical model
that describes the system under consideration, where the existence and
uniqueness of the solution are assumed. \eqref{eq:out_eq} is the
observation model. As the various types of mathematical models existing in
practice are described using polynomials \cite{Murray, glv}, we assume
that $f$ and $g$ are polynomial vectors in terms of $x, u$. These models
are defined for a duration $[T_0, T_1]$, where $T_0, T_1 \in \mathbb{R}$
and $T_0 < T_1$.% are time points that initial and final measurements are
taken respectively.
%As mentioned in Section \ref{sec:intro}, we focus on "unidentifiable" statespace
models, which we will define later.
\begin{comment}
\begin{remark}\label{remark:modelerr}
% As is clear from \eqref{eq:math_eq}, only continuous models are
concerned throughout this paper.
Some of the modelling approaches allow modelling errors such as
disturbances in \eqref{eq:math_eq} \cite{nonlinearsys}; however, we assume
that there are no such errors in model \eqref{eq:math_eq} at least in this
study, following practical situations. There are various models that do
not deal with such errors explicitly in estimating their unknown
parameters as in, e.g., \cite{Murray, science, pbpk, meshkatn}. We explain
the measurement noise later.
\end{remark}
\end{comment}
Before presenting the details of the time-series data to be dealt with,
the concept of model identifiability \eqref{eq:math_eq}, \eqref{eq:out_eq}
is described in short. As reviewed in
\cite{miao}, there are various definition of identifiability depending on
the given models and data assumptions; they can be are mainly categorised
into two types: structural and practical identifiability. These two types
involve different assumptions on data, with and without noise. As we first
investigate the theoretical aspects of the models, the structural ones are
of interest here. For simplicity, identifiable indicates structurally
identifiable in the following. For example, the definition of
unidentifiable models introduced in Definition 2.5 in \cite{miao} is as
follows: 
\begin{definition}\label{dfn:unidentifiable}
Model \eqref{eq:math_eq}, \eqref{eq:out_eq} is unidentifiable, if for any
input $u(t)$, there exist two parameters $a_\alpha \neq a_\beta$ in
parameter space $A$ such that $y(u; a_\alpha) = y(u; a_\beta)$ holds.
If model \eqref{eq:math_eq}, \eqref{eq:out_eq} is not unidentifiable, it
is called identifiable.
\end{definition}

Refer to \cite{miao} for the details.
Generally, model identifiability is investigated a priori, before
estimating the parameters. After the model is revealed to be identifiable,
a system described by the model can be naively investigated, given well-estimated
parameters. However, if the model is unidentifiable, system
investigation becomes difficult due to the non-uniqueness of the feasible
parameters. Therefore, if a model is shown to be unidentifiable,
modifications on the model are at times recommended in order
%things to do
to avoid non-uniqueness \cite{miao, meshkatl, zvi,chappell, evans}. %Thismay be rather suitable especially when model modifications are impossible
However, it is impossible or undesirable to modify models or data
measurements in certain applications, for example, if physiological-based
models \cite{pbpk} are used. Thus, we adopt a different approach in this
study; instead of checking the identifiability, which is one of properties
of the model, the model parameters complying with the observed data %,i.e., parameters that satisfy \eqref{eq:unidentifiable},
are investigated and therefore, enumerated comprehensively.

Now, we describe the assumption on the time series data used in this
study, in detail. We assume that the output data are measured only at a
finite number of time points in $[T_0, T_1]$, which differs from the
assumption generally applied in structural identifiability analysis. This
is because it may be difficult to obtain continuous measurements in
practice. Our purpose is to investigate the parameters of
\eqref{eq:math_eq}, \eqref{eq:out_eq} such that its inputs and output are
coincident with the observed data. However, in order to specify such
parameters under an algebraic framework, information on the higher
derivatives of the inputs and output of the model is required.
Accordingly, we also assume that a set of time derivatives of the input
and output data are obtained, particularly for constructing the
theoretical basis for our approach.
Note that this assumption can be relaxed as explained in
%Section \ref{sec:derivative_appprox}
Section \ref{sec:derivative_appprox}, in practice. More precisely, we investigate
the feasible parameters of model \eqref{eq:math_eq}, \eqref{eq:out_eq}
under the assumption that a set of derivatives of the input and output
data exist at $t_i \in [T_0, T_1]$.
\begin{align}
\hu &\in \mathbb{R}^{(M\times L)\times (N+1)},
\quad \hu= \left(\hu_1, \ldots, \hu_M \right)^T, \hu_m(i,j) =
\hu_m^{(j)}(t_i), m = 1,\ldots, M, \label{eq:in_data}\\
\hy &\in \mathbb{R}^{L\times (N+1)},\quad \hy(i,j) =
\hy^{(j)}(t_i)\label{eq:out_data}
\end{align}
%and derivatives of input data
%\begin{align}\label{eq:in_data}
%U_\d(t) = \left(u_\d(t), \ldots, u_\d^{(N)}(t) \right)^T\in
%\mathbb{R}^{M\times (N+1)}
%\end{align}
%over $[T_0, T_1]$
are observed. % denoted as $U_\d$, is observed over $[T_0, T_1]$.
%\begin{remark}
$L$ %appeared
in \eqref{eq:in_data}, \eqref{eq:out_data} indicates the number of
required measured time points in order to specify the parameter variety,
which is defined later. $L$ is determined by examining \eqref{eq:math_eq},
\eqref{eq:out_eq}, which surprisingly has a theoretical bound. See
%\ref{remark:K}
Remark \ref{remark:K} for details on $L$.
%\end{remark}
\begin{definition}\label{dfn:N-pv}
Suppose that time-series data \eqref{eq:in_data}, \eqref{eq:out_data}
corresponding to the input and output of state-space model
\eqref{eq:math_eq}, \eqref{eq:out_eq} are observed:
\begin{align}
\label{eq:model_data}
%\begin{aligned}
U_\d(t), t_i \in [T_0, T_1],&\quad {u}^{(j)}(t_i) = \left( \hu_1(i,j),\ldots,
\hu_M(i,j)\right)^T, y^{(j)}(t_i) = \hy(i,j) 
%\end{aligned},
\end{align}
where $i = 1,\ldots, L, j = 1,\ldots, N$.
We define the parameter variety
%"the set parameters
of the model given data $\hu$ and $\hy$ as the set of parameters with
which the model provides $u$ and $y$ coincident with the data.
\end{definition}
There appear to be no derivatives for $u, y$ in \eqref{eq:math_eq},
\eqref{eq:out_eq}; however, they exist in the differentials of the model,
which are demonstrated to be necessary for determining the parameter
variety. See %Section \ref{sec:formulation}
Section \ref{sec:formulation}, particularly %\ref{dfn:N-pv2}
Definition \ref{dfn:N-pv2}, for a
more precise definition. In this study, we extract the parameter variety
based on the framework of commutative algebra and algebraic geometry, and
thereby, show that it generally forms an algebraic variety.

\section{Proposed method}\label{sec:method}% and its significance in modelling approach}
In this section, a method for extracting the parameter variety of \eqref{eq:math_eq}, \eqref{eq:out_eq}, given data \eqref{eq:in_data}, \eqref{eq:out_data}, is proposed.% with an emphasis on the significance in its application for analysis of systems described by the model given the data.

\begin{remark}\label{remark:modeller}
Our method can be beneficial for modellers, who may not familiar with %differntial algebra nor
commutative algebra or algebraic geometry as illustrated in %Section \ref{sec:app}
Section \ref{sec:app}. Therefore, this section is organized for enabling readers to focus on the specific usage of our method. See %Section \ref{sec:theory}
Section \ref{sec:theory} for the theoretical details on the method.
 \end{remark}

The parameter variety of state-space model \eqref{eq:math_eq}, \eqref{eq:out_eq}, given time-series data \eqref{eq:in_data}, \eqref{eq:out_data} is investigated using Algorithm \ref{algorithm:alg1}, which outputs an algebraic variety containing the parameter variety. The algorithm involves two main steps:
\begin{enumerate}
    \item In step 1, \eqref{eq:math_eq} and \eqref{eq:out_eq} are manipulated in order to inspect the input-output relationships of the models. Thus, by eliminating the state variables, a certain equation that represents the input-output relationships is extracted from \eqref{eq:math_eq}, \eqref{eq:out_eq}.
    \item In steps 2 and 3, \eqref{eq:in_data} and \eqref{eq:out_data} are assigned to the equation obtained in step 1; thereby, an algebraic variety containing the parameter variety is extracted.
\end{enumerate}

\begin{remark}
Step 1 of Algorithm \ref{algorithm:alg1}, i.e., computation of an equation representing the input-output relationships of \eqref{eq:math_eq}, \eqref{eq:out_eq}, is the same as the one applied in structural identifiability analysis based on differential algebra, such as \cite{ljung1994}, \cite{meshkatn}, \cite{Saccomani_2003}. However, for the understanding of the reader, it is explained as follows.
\end{remark}

\begin{remark}
Technically, equations that represent the input-output relationships can be derived as a subset of the characteristic set for the differential ideal corresponding to \eqref{eq:math_eq}, \eqref{eq:out_eq} as explained in literature such as \cite{ljung1994} and \cite{Saccomani_2003}. On the other hand, \cite{forsman} presents a method for deriving such equations using the Gr\"{o}bner basis instead of the characteristic set, which is utilized in, for example, \cite{meshkatn}, \cite{Harrington_2016}, and \cite{ijrr2019}. We utilize the latter because it is sufficient for discussing the finite order of the derivatives. See Section %Section \ref{sec:formulation}
Section \ref{sec:formulation} for details on the relationship between differential algebra and our approach.
\end{remark}

\begin{algorithm}
{Derivation of an algebraic variety containing the parameter variety of \eqref{eq:math_eq}, \eqref{eq:out_eq}, given \eqref{eq:in_data}, \eqref{eq:out_data}}
\label{algorithm:alg1}
\begin{algorithmic}
\STATE{Input: Model \eqref{eq:math_eq}, \eqref{eq:out_eq}, time-series data \eqref{eq:in_data},\eqref{eq:out_data}}
\STATE{Output: Algebraic variety containing the parameter variety of \eqref{eq:math_eq}, \eqref{eq:out_eq} given \eqref{eq:in_data},\eqref{eq:out_data}}
\STATE{Step 1: Compute the basis of the input-output equations contained in the reduced Gr\"obner basis for the ideal generated by the sufficiently differentiated \eqref{eq:math_eq}, \eqref{eq:out_eq}}.
\STATE{Step 2: For each $t_i$, substitute data $\hu(t_i), \hy(i,j) (j = 1,\ldots, N)$ in input-output equation
\begin{align*}
    \Sigma_{l = 1}^L c_l(a)f_l(u,y,\dot{u}, \dot{y},\ldots, {u}^{(N)}, y^{(N)}) = f_{L+1}(u,y,\dot{u}, \dot{y},\ldots, {u}^{(N)}, y^{(N)}).
\end{align*}
$c_l(a)$ is the rational function of the parameters and $f_l(u,y,\dot{u}, \dot{y},\ldots, {u}^{(N)}, y^{(N)}) (l = 1,\ldots, L+1)$ is the differential monomial of $u, y$.}
\STATE{Step 3: Compute the values of the coefficients based on the input-output equation, denoted by $v_l (l=1,\ldots, L)$, by solving the linear simultaneous equations generated in the previous step}
\RETURN $\left\{ a \mid c_l(a) = v_l , l = 1,\ldots, L\right\}$
\end{algorithmic}
\end{algorithm}

\subsection{Step 1: Computation of the basis of the input-output equations of the model}\label{sec:step1}
In the field of control, the input-output relationships of linear state-space models are derived through methods such as the Laplace transformations. Thus, the transfer functions, which represent their input-output relationships, are investigated. As the transfer functions are uniquely determined from the models and represent the input-output behaviour, they are the main tools for inspecting the system properties in system theory. Through these functions, the input-output behaviour of models can be investigated because they do not contain information on the state variables.

As we deal with polynomial models \eqref{eq:math_eq}, \eqref{eq:out_eq}, which are basically nonlinear, algebraic manipulations are of concern, which is feasible for practical models using computational algebra. In particular, %due to derivatives of variables,
we introduce a commutative algebraic framework, which has been presented in \cite{forsman}, \cite{jirstrand}.

Under the algebraic framework, model \eqref{eq:math_eq}, \eqref{eq:out_eq} is regarded as a set of polynomials, namely, the polynomials of $x,u,y$ and their derivatives, such as
\begin{align}
    \dot{x}_1 -f_1(x,u;a), \ldots, \dot{x}_N -f_N(x,u;a), y-g(x,u;a). %\textcolor{red}{eq:generator}
    \label{eq:generator}
\end{align}
%\textcolor{red}{label{eqgenerator}}
Approximately, the ideal generated by \eqref{eq:generator}, % denoted as $\Sigma$,
is a set of polynomials obtained through the addition and multiplication of \eqref{eq:generator}. Unfortunately, elimination of the state variables and their derivatives may not be possible through the algebraic manipulation of \eqref{eq:generator}. As an example, we consider a model constituted by \eqref{eq:lv1}, where $x_2$ and $x_1$ are replaced by input $u$ and state variable $x$ and \eqref{eq:lvout}. Representing the model using polynomials, we obtain 
\begin{align}
    \frac{\d x}{\d t} - a_1x\left(1-\frac{x}{a_2}-a_3\frac{u}{a_2}\right), y - x. \label{eq:simple_competition}
\end{align}
Obviously, it is impossible to eliminate ${\d x}/{\d t}$ from \eqref{eq:simple_competition}. However, considering $\d y/\d t-\d x/\d t$, which is the derivative of the second polynomial in \eqref{eq:simple_competition}, both $x$ and $\d x/\d t$ can be eliminated:
\begin{align*}
    \frac{\d y}{\d t} - a_1y\left(1-\frac{y}{a_2}-a_3\frac{u}{a_2}\right).
\end{align*}
%\begin{align*}
%    \frac{\d^2 x}{\d t^2} - a_1\frac{\d x}{\d t}\left(1-\frac{x}{a_2}-a_3\frac{u}{a_2}\right) -a_1x\left(-\frac{1}{a_2}\frac{\d x}{\d t}-a_3\frac{1}{a_2}\frac{\d u}{\d t}\right), \frac{\d y}{\d t} - \frac{\d x}{\d t}. \label{eq:simple_competition}
%\end{align*}
The number of differentiations of \eqref{eq:simple_competition} needed to eliminate the state variable derivatives is bounded \cite{forsman}, \cite{jirstrand}, \cite{meshkatn}.
Hence, if we have \eqref{eq:math_eq}, \eqref{eq:out_eq} and their derivatives up to the required minimum, the state variables and their derivatives can be eliminated by algebraic manipulation. We refer to these equations that do not include state variables nor their derivatives as the input-output equations of \eqref{eq:math_eq}, \eqref{eq:out_eq}. Although such manipulations are not unique for elimination, it is possible to find a finite set of polynomials that enumerate all the input-output equations through %Gr\"obner
Gr\"{o}bner basis computations. The set is a subset of the Gr\"{o}bner basis without the state variables and their derivatives for the ideal generated by \eqref{eq:math_eq}, \eqref{eq:out_eq} and their derivatives up to required minimum. Therefore, the set can be regarded as the basis for all the input-output equations. According to \cite{forsman, jirstrand}, the set contains only one polynomial, which we refer to as the basis of the input-output equations of \eqref{eq:math_eq}, \eqref{eq:out_eq}. This basis can be computed using %\ref{algorithm:alg2}
Algorithm \ref{algorithm:alg2} proposed in \cite{forsman, jirstrand}. Termination of the algorithm is guaranteed.
\begin{algorithm}
%\caption{Derivation of the basis of the input-output equations of \eqref{eq:math_eq}, \eqref{eq:out_eq} and the minimal number of required differentiation}
Derivation of the basis of the input-output equations of \eqref{eq:math_eq}, \eqref{eq:out_eq} and the minimal number of required differentiation.
\label{algorithm:alg2}
\begin{algorithmic}
\STATE{Input: Model \eqref{eq:math_eq}, \eqref{eq:out_eq}}
\STATE{Output: Basis of the input-output equations of \eqref{eq:math_eq}, \eqref{eq:out_eq}}
\STATE $S := \phi, i := 1$
\WHILE{$S = \phi$}
\STATE{Let $I_i$ denote the ideal over polynomial ring $\mathbb{C}(a)[x,\ldots, x^{(i)}, y,\ldots, y^{(i)}, u, \ldots, u^{(i-1)}]$ that is generated the derivatives of \eqref{eq:math_eq}, \eqref{eq:out_eq} up to the $i$-th order}.
\STATE{Fix the monomial order defined on the ring as the lexicographic order for which the order of the variables is $x^{(i)} < \cdots x < \cdots < y^{(i)} < \cdots < y < u^{(i-1)} < \cdots < u$.}
\STATE{Compute the reduced Gr\"obner basis for $I_i$.}
\STATE{$S :=$ elements of the Gr\"obner basis without state variables and their derivatives.}
\STATE{$i := i + 1$}
\ENDWHILE{}
\STATE{$l := i$}
%\STATE{Divide $S$ by the coefficient of the monomial with the highest degree of the highest order derivative of of $y$ and let this be $S$}
\RETURN $S, l \leq N$
\end{algorithmic}
\end{algorithm}

\begin{remark}
In general, the Gr\"{o}bner basis for the ideal over the polynomial ring with respect to the monomial order is not unique. However, the reduced Gr\"{o}bner basis, which is a special type of Gr\"{o}bner basis, of the ideal is unique with respect to the order. Therefore, the reduced Gr\"{o}bner basis is computed to obtain the basis of the input-output equations as follows. See %Section \ref{sec:main}
Section \ref{sec:main} for the definition of the Gr\"{o}bner basis and the reduced Gr\"{o}bner basis.
\end{remark}

\begin{comment}
\begin{remark}
Division in the last step of Algorithm \ref{algorithm:alg2} is performed in order to obtain the basis of the input-output equations containing monomials, whose coefficients are constants. Through this operation, the basis of the input-output equation is forced to be unique, which is similar to that in \cite{Saccomani_2003, meshkatn}. \textcolor{red}{(division by parameters)}
\end{remark}
\end{comment}

In general, Gr\"obner basis computations can be performed using computer algebra software such as Singular, Mathematica, Magma, and Maple. An example of a series of Singular commands for computing the basis of the input-output equations of a model using %\ref{algorithm:alg2}
Algorithm \ref{algorithm:alg2} is shown below:

%Characteristic sets for $\Sigma$, which is a certain finite set of differential polynomials,
%characterize $\Sigma$ in a sense that they determine whether a given differential polynomial belongs to the ideal or not. Thanks to such property of characteristic sets for \eqref{eq:math_eq}, \eqref{eq:out_eq}, there is no longer need to concern dependency of, for example, eliminations of state variables on the way of manipulations of \eqref{eq:generator}.
%We call equations corresponding to differential polynomials in the characteristic set without state variables and their derivatives the input-output equations.

%Characteristic set for \eqref{eq:math_eq}, \eqref{eq:out_eq} can be computed by computer algebra software such as Maple \cite{Harrington_2016}, \cite{Boulier_de}.

\begin{example}
In this example, we consider \eqref{eq:lv1}, \eqref{eq:lv2}, \eqref{eq:lvout} and derive the basis of the input-output equations of the model; the corresponding Singular commands needed are shown. %An input-output equation of the model is derived via computation of the Gr\"obner basis for the ideal generated by a sufficiently differentiated model. Although characteristic sets of $\Sigma$ are not derived directly, the input-output equation contained in the set is obtained by this method. In order to compute the input-output equatoin, ones investigate how many times the model needs to be differentiated. According to \cite{forsman}, such number of differentiations is bounded by the number of state variables, and thus, it is guaranteed that the input-output equations can be obtained. %Technically speaking, this method applies non-differential algebra.

First, \eqref{eq:lvout} is differentiated:
\begin{align}
    \frac{\d }{\d t}y &= \frac{\d }{\d t}x_1. \label{eq:lvout1}
\end{align}
%Observing that the right hand side of the above only has the first derivative of $x_1$,
We then check whether the state variables and their derivatives can be eliminated from the ideal generated by
\begin{align}
      \frac{\d x_1}{\d t} = a_1x_1\left(1-\frac{x_1}{a_2}-a_3\frac{x_2}{a_2}\right), \frac{\d x_2}{\d t} = a_4x_2\left(1-\frac{x_2}{a_5}-a_6\frac{x_1}{a_5}\right),
    y = x_1,
        \frac{\d }{\d t}y = \frac{\d }{\d t}x_1 \label{eq:i1},
\end{align}
where \eqref{eq:lvout1} is added to \eqref{eq:lv1}, \eqref{eq:lv2}, and \eqref{eq:lvout}. The reduced Gr\"obner basis for the ideal generated by \eqref{eq:i1} with respect to lexicographic ordering $\dot{x}_2 > \dot{x}_1 > x_2 > x_1 > \dot{y} > y$ is computed as follows:
%At this time, the corresponding commands if ones use Singluar are shown.
\begin{lstlisting}[basicstyle=\ttfamily\footnotesize, frame=single]
ring r = (0,a1, a2, a3, a4, a5, a6), (x4, x3, x2, x1, y1, y0),lp;
ideal i1 = x3 + a1*x1*(x1/a2 + (a3*x2)/a2 - 1),
x4 + a4*x2*(x2/a5 + (a6*x1)/a5 - 1), y0 - x1, y1 - x3;
option(redSB);
groebner(i1);
\end{lstlisting}
\begin{comment}
> ring r = (0,a1, a2, a3, a4, a5, a6), (x4, x3, x2, x1, y1, y0),lp;
> ideal i1 = x3 + a1*x1*(x1/a2 + (a3*x2)/a2 - 1),
. x4 + a4*x2*(x2/a5 + (a6*x1)/a5 - 1), y0 - x1, y1 - x3;
> groebner(i1);
_[1]=x1-y0
_[2]=(a1*a3)*x2*y0+(a2)*y1+(a1)*y0^2+(-a1*a2)*y0
_[3]=x3-y1
_[4]=(a5)*x4+(a4)*x2^2+(a4*a6)*x2*y0+(-a4*a5)*x2
> option(redSB);
> groebner(i1);
_[1]=x1-y0
_[2]=(a1*a3)*x2*y0+(a2)*y1+(a1)*y0^2+(-a1*a2)*y0
_[3]=x3-y1
_[4]=(a1*a3*a5)*x4+(a1*a3*a4)*x2^2+(-a1*a3*a4*a5)*x2+(-a2*a4*a6)*y1+(-a1*a4*a6)*y0^2+(a1*a2*a4*a6)*y0.
\end{comment}
Note that $\dot{x}_2, \dot{x}_1, x_2, x_1, \dot{y}, y$ are considered non-differentiated variables and denoted by ``x4", ``x3", ``x2", ``x1", ``y1", and ``y0" in the above commands. ``Option(redSB)" forces us to compute the reduced basis. We obtain the following as the output:
\begin{lstlisting}[basicstyle=\ttfamily\footnotesize, frame=single]
_[1]=x1-y0
_[2]=(a1*a3)*x2*y0+(a2)*y1+(a1)*y0^2+(-a1*a2)*y0
_[3]=x3-y1
_[4]=(a1*a3*a5)*x4+(a1*a3*a4)*x2^2+(-a1*a3*a4*a5)*x2
+(-a2*a4*a6)*y1+(-a1*a4*a6)*y0^2+(a1*a2*a4*a6)*y0.
\end{lstlisting}
Obviously, there are no elements that do not contain state variables or their derivatives. Hence, as the second step, \eqref{eq:lvout1} is differentiated again.
%:
%\begin{align}
%    \frac{\d^2 }{\d t^2}y &= \frac{\d^2 }{\d t^2}x_1. \label{eq:lvout2}
%\end{align}
%Since $\frac{\d^2 }{\d t^2}x_1$ is not in the original model, it is formed by differentiating \eqref{eq:lv1}, which also gives $\frac{\d }{\d t}x_2$. Thus, we check whether state variables and their derivatives can be eliminated from
We then check whether the state variables and their derivatives can be eliminated from
\begin{align}
      \frac{\d x_1}{\d t} &= a_1x_1\left(1-\frac{x_1}{a_2}-a_3\frac{x_2}{a_2}\right), \frac{\d x_2}{\d t} = a_4x_2\left(1-\frac{x_2}{a_5}-a_6\frac{x_1}{a_5}\right),\nonumber \\ \frac{\d^2 x_1}{\d t^2} &= a_1\frac{\d}{\d t}x_1\left(1-\frac{x_1}{a_2}-a_3\frac{x_2}{a_2}\right) %a_1x_1\frac{\d}{\d t}\left(1-\frac{x_1}{a_2}-a_3\frac{x_2}{a_2}\right)\\
      -a_1x_1\left(\frac{1}{a_2}\frac{\d}{\d t}x_1+ \frac{a_3}{a_2}\frac{\d}{\d t}x_2\right),\nonumber \\
      \frac{\d^2 x_2}{\d t^2} &= a_4\frac{\d}{\d t}x_2\left(1-\frac{x_2}{a_5}-a_6\frac{x_1}{a_5}\right) %a_1x_1\frac{\d}{\d t}\left(1-\frac{x_1}{a_2}-a_3\frac{x_2}{a_2}\right)\\
      -a_4x_2\left(\frac{1}{a_5}\frac{\d}{\d t}x_2+ \frac{a_6}{a_5}\frac{\d}{\d t}x_1\right),\nonumber\\
    y &= x_1,
        \frac{\d }{\d t}y = \frac{\d }{\d t}x_1, \frac{\d^2 }{\d t^2}y - \frac{\d^2 }{\d t^2}x_1,\label{eq:competition_ideal}
\end{align}
in the same manner as before. By computing the reduced Gr\"obner basis for the ideal generated by the above equation with respect to lexicographic ordering $\ddot{x}_2> \ddot{x}_1> \dot{x}_2 > \dot{x}_1 > x_2 > x_1 >\ddot{y}> \dot{y} > y$, we obtain an element of the reduced Gr\"obner basis that does not contain a state variable nor its derivative assuming that $a_1, a_2, a_3, a_5 \neq 0$.%, which is the basis of the input-output equations of \eqref{eq:lv1},\eqref{eq:lv2},\eqref{eq:lvout}
:
\begin{align}
    &\frac{-a_1a_2a_3a_5-a_2^2a_4}{a_1a_2a_3a_5}{\left(\frac{\d }{\d t}y\right)}^2 +\frac{a_1^2a_3a_5+a_1a_2a_3a_4a_6-2a_1a_2a_4}{a_1a_2a_3a_5}\left(\frac{\d}{\d t}y\right)y^2 \nonumber \\
    +& \frac{2a_1a_2^2a_4-a_1a_2a_3a_4a_5}{a_1a_2a_3a_5}\left(\frac{\d}{\d t}y\right)y +\frac{a_1^2a_3a_4a_6-a_1^2a_4}{a_1a_2a_3a_5}y^4 \nonumber \\
    +&\frac{-a_1^2a_2a_3a_4a_6+2a_1^2a_2a_4-a_1^2a_3a_4a_5}{a_1a_2a_3a_5}y^3 +\frac{-a_1^2a_2^2a_4+a_1^2a_2a_3a_4a_5}{a_1a_2a_3a_5}y^2 \nonumber \\
    =& -\left(\frac{\d^2}{\d t^2}y\right)y \label{eq:ex2_ioeq}.
\end{align}
This is the basis of the input-output equations of \eqref{eq:lv1}, \eqref{eq:lv2}, \eqref{eq:lvout}. The minimal number of differentiations required, $L$, is revealed to be two, which is the same as the number of state variables in this case.
Finally, dividing by $a_1a_2a_3a_5$, which is the coefficient of the monomial containing the highest degree of $\d^2 y/\d t^2$, we obtain the basis of the input-output equations of \eqref{eq:lv1}, \eqref{eq:lv2}, \eqref{eq:lvout}, uniquely.
\begin{comment}
\begin{align*}
    &\left(\frac{\d^2}{\d t^2}y\right)y+\frac{-a_1a_2a_3a_5-a_2^2a_4}{a_1a_2a_3a_5}{\left(\frac{\d }{\d t}y\right)}^2 \nonumber \\
    &+\frac{a_1^2a_3a_5+a_1a_2a_3a_4a_6-2a_1a_2a_4}{a_1a_2a_3a_5}\left(\frac{\d}{\d t}y\right)y^2 + \frac{2a_1a_2^2a_4-a_1a_2a_3a_4a_5}{a_1a_2a_3a_5}\left(\frac{\d}{\d t}y\right)y\nonumber \\
    &+\frac{a_1^2a_3a_4a_6-a_1^2a_4}{a_1a_2a_3a_5}y^4 +\frac{-a_1^2a_2a_3a_4a_6+2a_1^2a_2a_4-a_1^2a_3a_4a_5}{a_1a_2a_3a_5}y^3\nonumber \\
    &+\frac{-a_1^2a_2^2a_4+a_1^2a_2a_3a_4a_5}{a_1a_2a_3a_5}y^2 = 0.
\end{align*}
\end{comment}
\end{example}

\subsection{Steps 2 and 3: Extraction of the algebraic variety containing the parameter variety}\label{sec:step2to3}
The basis of the input-output equations of \eqref{eq:math_eq}, \eqref{eq:out_eq} is, in general, as follows:
\begin{align}\label{eq:ioeq}
    \Sigma_{l = 1}^L c_l(a)f_l(u,y,\dot{u}, \dot{y},\ldots, {u}^{(L-1)}, y^{(L)}) = f_{l+1}(u,y,\dot{u}, \dot{y},\ldots, {u}^{(L-1)}, y^{(L)}),
\end{align}
where $c_l(a)$ is the rational function of the parameters and $f_l(u,y,\dot{u}, \dot{y},\ldots, {u}^{(L-1)}, y^{(L)}) (l = 1,\ldots, L)$ is a monomial of $y, \ldots, y^{(L)}, u, \ldots u^{(L-1)}$. $f_{L+1}(u,y,\dot{u}, \dot{y},\ldots, {u}^{(L-1)}, y^{(L)})$ denotes the polynomials of $y, \ldots, y^{(L)}, u, \ldots u^{(L-1)}$ with constant coefficients. As \eqref{eq:ioeq} is derived from \eqref{eq:math_eq}, \eqref{eq:out_eq}, %without considering specific time points,
it must be satisfied over $[T_0, T_1]$.

\begin{comment}
\begin{example}
The basis of the input-output equation after division in the last step of Algorithm \ref{algorithm:alg2} is as follows:
\begin{align*}
         &\frac{(-a_1a_2a_3a_5-a_2^2a_4)}{a_1a_2a_3a_5}{\left(\frac{\d }{\d t}y\right)}^2\nonumber \\
    &+\frac{(a_1^2a_3a_5+a_1a_2a_3a_4a_6-2a_1a_2a_4)}{a_1a_2a_3a_5}\left(\frac{\d}{\d t}y\right)y^2 + \frac{(2a_1a_2^2a_4-a_1a_2a_3a_4a_5)}{a_1a_2a_3a_5}\left(\frac{\d}{\d t}y\right)y\nonumber \\
    &+\frac{(a_1^2a_3a_4a_6-a_1^2a_4)}{a_1a_2a_3a_5}y^4 +\frac{(-a_1^2a_2a_3a_4a_6+2a_1^2a_2a_4-a_1^2a_3a_4a_5)}{a_1a_2a_3a_5}y^3\nonumber \\
    &+\frac{(-a_1^2a_2^2a_4+a_1^2a_2a_3a_4a_5)}{a_1a_2a_3a_5}y^2 = -\left(\frac{\d^2}{\d t^2}y\right)y, \quad a_1a_2a_3a_5 \neq 0.
\end{align*}
\end{example}
\end{comment}

 %Since the %If there is no differential monomial with constants, it is rearranged into the monic one as conducted in \cite{Mishra}, \cite{meshkatn}. More specifically, once every coefficient is devided by the coefficient of the monomial containing the highest order of highest derivatives $y^{(j)}$, $f_{K+1}(u,y,\dot{u}, \dot{y},\ldots, {u}^{(N)}, y^{(N)})$ becomes a non-zero polynomial.

In step 2, substituting \eqref{eq:in_data}, \eqref{eq:out_data} in \eqref{eq:ioeq} for each time point $t_i(i = 1,\ldots, K)$, simultaneous linear equations are obtained:
\begin{align}\label{eq:sys_ioeq}
&\left(
\begin{array}{ccc}
f_1(\hu(t_1), \hy(1,:))& \cdots& f_{L}(\hu(t_1), \hy(L,:))\\
\vdots& \ddots & \vdots\\
f_1(\hu(t_1), \hy(L,:))& \cdots & f_{L}(\hu(t_1), \hy(L,:))\\
\end{array}
\right) \left(
\begin{array}{c}
c_1(a)\\
\vdots \\
c_{L}(a)\\
\end{array}
\right) %\nonumber \\
%&
= \left(\begin{array}{c}
f_{L+1}(\hu(t_1), \hy(1,:))\\
\vdots \\
f_{L+1}(\hu(t_L), \hy(L,:))\\
\end{array}\right).
\end{align}
In practice, the derivatives of the time-series data may not be observed. Besides, the data may be noisy due to measurement errors. Thus, theoretically, it suffices to substitute \eqref{eq:in_data}, \eqref{eq:out_data} in \eqref{eq:ioeq}; however, when dealing with actual data, pre-processing of the data may be required. Hence, in %Section \ref{sec:derivative_appprox}
Section \ref{sec:derivative_appprox}, data pre-processing %an explanation on estimation of dervatives of output data where
is discussed assuming the input data are represented by sufficiently smooth functions.
\begin{remark}\label{remark:K}
Using time-series data observed at more than $L$ time points, simultaneous linear equations such as \eqref{eq:sys_ioeq} can be formed as well, although it may be better to do so in terms of the robustness against data noise. %, we leave this as future works.
In addition, we note that data need to be measured for at least $L$ time points, which is the number of monomials in \eqref{eq:ioeq} with non-constant coefficients in order to obtain a unique solution for \eqref{eq:sys_ioeq}. This restriction is not applicable, if data pre-processing is performed as explained in Section %Section \ref{sec:derivative_appprox}
Section \ref{sec:derivative_appprox}.
\end{remark}

In step 3, the above is solved in terms of $c_k(a)$, for which the solution is denoted by $v_k (k=1,\ldots, L)$.
%\begin{remark}
In this study, we assume that \eqref{eq:sys_ioeq} admits a unique solution.
%\end{remark}
Approximately, as \eqref{eq:ioeq} is the equation that the inputs, output, and derivatives of \eqref{eq:math_eq}, \eqref{eq:out_eq}, given the parameters, satisfy at least for the measured time points, the parameters of the parameter variety of \eqref{eq:math_eq}
, \eqref{eq:out_eq}, given \eqref{eq:in_data} ,\eqref{eq:out_data}, satisfy $\left\{ c_l(a)= v_l , l = 1,\ldots, L\right\}$. Thus, the output of Algorithm \ref{algorithm:alg1} is obtained:
\begin{align}
    \left\{ a \mid c_l(a) = v_l , l = 1,\ldots, L\right\},\label{eq:algebraicvariety}
\end{align} which contains the parameter variety. Although \eqref{eq:algebraicvariety} may include parameters that are not in the parameter variety, it is possible to extract such parameters, as shown in %\ref{thm:extention}
Theorem \ref{thm:extention} and demonstrated through an example in %\ref{ex:singular}
Example \ref{ex:singular}. Hence, the parameters in \eqref{eq:algebraicvariety} are generally in the parameter variety of \eqref{eq:math_eq}, \eqref{eq:out_eq}, given \eqref{eq:in_data}, \eqref{eq:out_data}. See Section %Section \ref{sec:theory}
Section \ref{sec:theory} for the theoretical justification of the proposed method.

\subsection{Data pre-processing %an explanation on estimation of dervatives of output data where
given input data as sufficiently smooth functions}\label{sec:derivative_appprox}%(If derivatives of data is nor observed)
%\textcolor{black}{
%It should be noted that the derivatives of input data \eqref{eq:in_data} coincident with verivatives of inputs where the inputs are assumed to be smooth enough over $[T_0, T_1]$. Hence,
There are two issues to consider in dealing with real-world data using the proposed method: measurement errors and derivatives of the data, which are not typically observed. Assuming the input data is given as smooth functions over $[T_0, T_1]$, practical methods to deal with the above mentioned issues are discussed below. The following two procedures are
considered for data pre-processing.

\begin{remark}
If input data are obtained at sufficient time points over $[T_0, T_1]$, they can be used to approximate a sufficiently smooth input function.  %Especially, if no input to \eqref{eq:math_eq}, \eqref{eq:out_eq} is given, which is a case shown in Section \ref{sec:app},
This suggests that the following data pre-processing would also work for such input data. %Estimations of derivatives of both input and output data are future works.
\end{remark}%}

As the actual data may include measurement noise, \eqref{eq:math_eq}, \eqref{eq:out_eq} typically does not fit the data perfectly. Therefore, the parameter variety of \eqref{eq:math_eq}, \eqref{eq:out_eq} given such data may not contain any parameters. On the other hand, in general, parameter estimations of \eqref{eq:math_eq}, \eqref{eq:out_eq} for such data are performed to decrease the residuals between such data and the model outputs as much as possible; thus, parameters that fit the data to a certain  extent are of interest \cite{Ramsay}. %Based on this, it may be better to obtain the parameter variety given pseudo data obtained by fitting \eqref{eq:math_eq}, \eqref{eq:out_eq} to the observed data \cite{Ramsay} which may not be exact due to measurement errors.
Therefore, it may be better to obtain the parameter variety given the input data and a trajectory of the model output fitted to the output data. We refer to this trajectory as pseudo-data, which may not exactly fit the output data at the measurement points due to measurement error.
As \eqref{eq:ioeq} describes the input-output relationships of the model, it can also be used for generating pseudo-data. The advantage of using \eqref{eq:ioeq} for pseudo-data generation is that the initial values of the state variables are not required. However, it may be difficult to solve \eqref{eq:ioeq} numerically, compared to the model in the state form, suggesting the application of \eqref{eq:math_eq}, \eqref{eq:out_eq} naively. The parameter variety of \eqref{eq:math_eq}, \eqref{eq:out_eq}, given pseudo-data, is guaranteed to have at least one parameter because %there exist a parameter that generates the pseudo data.
pseudo-data are generated by a parameter.

As data derivatives are not typically observed in practice, they need to be approximated. After obtaining pseudo-data as the numerical solution of \eqref{eq:math_eq}, \eqref{eq:out_eq} or \eqref{eq:ioeq}, higher derivatives of the pseudo-output data can be approximated through numerical differentiation methods using the pseudo-data. Note that pseudo-data at more than $L$ time points is guaranteed, if numerical computations of \eqref{eq:math_eq}, \eqref{eq:out_eq} or \eqref{eq:ioeq} are performed with a sufficiently small time-step size.

%%%%%%%%%%%%%%section A%%%%%%%%%

\section{Application of the proposed method in viral dynamics analysis}\label{sec:app}
In this section, we illustrate the working of the proposed method, through a viral dynamics analysis application example. After elucidating the model that describes the viral dynamics and the observed data, followed by parameter estimation through a conventional approach, we apply the proposed method, highlighting its significance in modelling.
\begin{remark}{
%As shown later, the output of Algorithm \ref{algorithm:alg1} may contain parameters that are not in the parameter variety of \eqref{eq:math_eq}, \eqref{eq:out_eq} given \eqref{eq:in_data},\eqref{eq:out_data}.
As mentioned in %Section \ref{sec:step2to3}
Section \ref{sec:step2to3}, the output of Algorithm \ref{algorithm:alg1} may contain parameters that are not included in the parameter variety. However, in order to emphatically illustrate our idea, i.e., the extraction of all the possible parameters, such atypical parameters are not considered in this section. We present a method to remove such parameters in %Section \ref{sec:theory}
Section \ref{sec:theory}, followed by an example using the same model in %\ref{ex:singular}
Example \ref{ex:singular}.
}
\end{remark}

\subsection{Model describing the viral dynamics and viral data}\label{sec:viralmodel}
In studies such as \cite{science}, the dynamics of the hepatitis C virus (HCV) during interferon-$\alpha$-2b(IFN) therapy is analysed using the following mathematical model:
%\begin{comment}
\begin{align}
\frac{\d x_1}{\d t} = a_1-a_2x_1-a_3x_1x_3, \frac{\d x_2}{\d t} = a_3x_1x_3-a_4x_2,
\frac{\d x_3}{\d t} = (1-a_5)a_6x_2-a_7x_3, \label{eq:virus0}
\end{align}
%\end{comment}
where $x_1, x_2, x_3$ are time $t$ dependent variables and $a_1, \ldots, a_7$ are time-invariant parameters. $x_3$ is measured as $y = x_3$.
More precisely, $x_1, x_2, x_3$ are the number of target cells, number of productively infected cells, and the viral load, respectively. Target cells are produced at rate $a_1$ and die at a constant death rate $a_2$. They become infected at a constant infection rate $a_3$. The infected cells produce HCV at an average rate of $a_6$ and die at a constant rate $a_4$ per cell per day. Virions are cleared at a constant clearance rate $a_7$. IFN reduces the production of virions from the infected cells by a fraction $(1-a_5)$, where $0 \leq a_5 \leq 1$. It is to be noted that $a_1 \geq 0, \ldots, a_7 \geq 0$. Further, based on \cite{science}, the model is rearranged assuming that the number of target cells $x_1$ is constant for at least two weeks. In addition, the value of $x_1(0)$ is assumed to be equal to the value of $x_1$ at the steady state of model \eqref{eq:virus0}, where $a_5 = 0$, indicating that IFN has no effect at $t < 0$. Based on quasi-steady state analysis, it is shown that $x_1(t) = \frac{a_4a_7}{a_3a_6}$ and $x_2(T_0) = (a_7/a_6)x_3(T_0), x_3(T_0)$, where $T_0$ denotes the delay in viral decay. For details on the steady state analysis, refer to \cite{science}. Thus, the rearranged model is obtained as follows:
\begin{align}
\frac{\d x_2}{\d t} &= \frac{a_4a_7}{a_6}x_3-a_4x_2, %\label{eq:m2x2}\\
\frac{\d x_3}{\d t} = (1-a_5)a_6x_2-a_7x_3, %\label{eq:m2x3}\\
y = x_3, \label{eq:virous}\\
x_2(T_0) &= (a_7/a_6)x_3(T_0)\label{eq:initcond},
\end{align}
where the parameters and variables correspond to those of model \eqref{eq:virus0}. This is defined over $[T_0, T_1]$, where $T_1 = 24\times 14 = 336$ h, at which the last measurement is made. Obviously, there is no input to this model. In the next subsection, the proposed method is applied to \eqref{eq:virous}, given the time-series data explained below. Although \eqref{eq:virous} is linear and therefore, appears simple, it should be noted that the proposed method is applicable to non-linear models because it is based on algebra.

In \cite{science}, the time-series data of the viral load is observed for three subjects, denoted by 1-H, 2-D, and 3-D, as shown in %\ref{fig:data}
Fig. \ref{fig:data} with specific symbols.
\begin{figure}
\centering
\includegraphics[width=0.7\textwidth, clip]{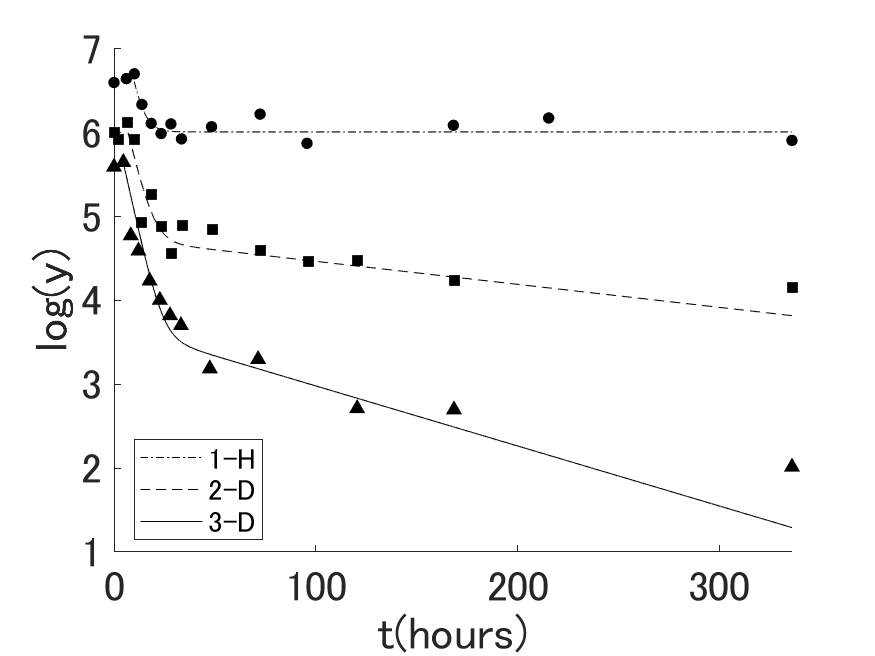}
\caption{Observed time-series data for three subjects (denoted by $1-\mathrm{H}$, $2-\mathrm{D}$, and $3-\mathrm{D}$), for whom $5\times10^6$, $10\times10^6$, and $15\times10^6$(IU) of IFN were administrated, respectively. The data were extracted from the images in \cite{science}. The horizontal-axis represents the time in date, whereas the vertical axis represents the logarithm of the viral load (HCV RNA copies /ml)
with base 10. The data are plotted using subject-specific symbols ($1-\mathrm{H}$:circle, $2-\mathrm{D}$:square, $3-\mathrm{D}$:triangle). %Each white symbol is the initial condition of the model \eqref{eq:virous} for each subject, which is estimated in \cite{science}.
The lines indicate the numerical solutions, given parameter $a_4,a_5, a_7$ shown in %\ref{tab:neumann_parameter}
Table \ref{tab:neumann_parameter} over $[T_0, T_1]$, where $T_1 = 24\times 14 = 336$ h. $T_0$ and $y(T_0)$ are assumed to be $(10,7,5)$ h and $(4.1, 1.0, 0.4)\times 10^6$ copies /ml for $1-H,2-D,and 3-D$, respectively according to \cite{science}.}
\label{fig:data}
\end{figure}
In the original study, $a_4,a_5, and a_7$ are estimated using the non-linear least squares methods. %based on Levenberg-Marquardt algorithm
$a_6$ is not estimated because it does not appear in the exact solution of $y = x_3$ as shown in %Section \ref{sec:exact}
Section \ref{sec:exact}. %The standard error of estimations are evaluated by a bootstrap method \cite{Efron_1986}. Unfortunately, the initial parameters for least square methods are not declaired in \cite{science}.
%100 different initial parameters and the averages of estimated parameters are computed, which are shown in \ref{tab:neumann_parameter}.
\begin{table}[htb]
 \caption{%Means and standard errors of
 Estimated parameters for $1-\mathrm{H}$, $2-\mathrm{D}$, and $3-\mathrm{D}$ in \cite{science}. %The standard errors in persentages are presented in brackets followed by the means.
 %See \cite{science} for the details of the estimations.
 %$E$ is the error between the model output and the observed data, which is defined as \eqref{eq:error}.
 }
\begin{center}
  \begin{tabular}{|c|c|c|c|c|c|c|c|c|c|c|c|} \hline
Subject &  $a_4$ & $a_5$ & $a_7$ \\%&$E$\\ \hline \hline %& $t_{\rm{c}}$ & $k_{\rm{a}}$ & $k_{\rm{e}}$ \\
\hline \hline
%1-H &0(0.01) &0.75(0.1)  &6.9(0.2)\\
%2-D& 0.16(0.04) &0.95(0.1)  &5.6(0.5)\\
%3-D&0.4(0.05) &0.99(0.2) &6(0.8) \\ \hline
1-H &0 &0.75  &6.9\\
2-D& 0.16 &0.95  &5.6\\
3-D&0.4 &0.99 &6 \\ \hline
\end{tabular} \label{tab:neumann_parameter}
   \end{center}
\end{table}
The numerical solutions, given the estimated parameters in Table %\ref{tab:neumann_parameter}
\ref{tab:neumann_parameter}, are denoted by lines in %\ref{fig:data}
Fig. \ref{fig:data}.

%\begin{remark}
%In numerical experiments of \cite{science}, the derivatives of time-series data is not observed nor concerned. On the other hand, our method needs them. Therefore, comparing estimated parameters is appearently not fair. However, the point is that the proposed method may be able to enhance results obtained through conventional method such as \cite{science}. For that reason in the next subsection the results obtained in \cite{science} are partially reused.
%\end{remark}

\subsection{Application of the proposed method to the viral model}\label{sec:viralmodel_pv}
In this section, the proposed method is applied to the model and data described in the previous section. Based on the methods described in Section \ref{sec:step1}, the basis of the input-output equations of \eqref{eq:virous} is derived as follows:
\begin{align}
&a_{4}a_5a_7y+\left(a_{4}+a_{7}\right)\frac{\d y}{\d t} = -\frac{\d^2 y}{\d t^2}. \label{eq:virous_ioeq}
\end{align}
 %By letting the coefficients of the equation \eqref{eq:virous_ioeq} be $d_1, d_2$ respectively, we obtain the manifold
%\begin{align}
%d_1 = a_4a_5a_7, \ d_2 = a_4 + a_7 \ %\nonumber \\
%\Leftrightarrow \ d_1 = a_4a_5(d_2-a_4). \label{eq:mf2}
%\end{align}
%Thus, the model is shown to be unidentifiable.
The values of the coefficients of \eqref{eq:virous_ioeq}, which are denoted by $v_1, v_2$, are uniquely determined by substituting output data $y_\d, \frac{\d y_\d}{\d t}, \frac{\d^2 y_\d}{\d t^2}$ at time points $t_1, t_2$ in equation \eqref{eq:virous_ioeq} as follows:
\begin{align}
  &\left\{
    \begin{array}{l}
      \left.v_1y(t_1)+v_2\left.\frac{\d y}{\d t}\right|_{t = t_1} = -\frac{\d^2 y}{\d t^2}\right|_{t = t_1} \\
     \left. v_1y(t_2)+v_2\left.\frac{\d y}{\d t}\right|_{t = t_2} = -\frac{\d^2 y}{\d t^2}\right|_{t = t_2}
    \end{array}
  \right. \nonumber \\
  &\Leftrightarrow \left(
    \begin{array}{c}
      v_1 \\
      v_2
    \end{array}
  \right) =
 % -{\left(
-\begin{pmatrix}
 %   \begin{array}{cc}
      y(t_1) & \left.\frac{\d y}{\d t}\right|_{t = t_1} \\
      y(t_2) & \left.\frac{\d y}{\d t}\right|_{t = t_2}
 %   \end{array}
    \end{pmatrix}^{-1}
  %\right)}^{-1}
  \left(
    \begin{array}{c}
      \left.\frac{\d^2 y}{\d t^2}\right|_{t = t_1} \\
      \left.\frac{\d^2 y}{\d t^2}\right|_{t = t_2}
    \end{array}
  \right), \label{eq:find_v}
\end{align}
where the above matrix is assumed to be regular. Thereby, $v_1,v_2$ can be estimated and the algebraic variety containing the parameter variety is obtained:
\begin{align}\label{eq:pv_virus}
    \left\{ a \mid a_4a_5a_7 = v_1, a_4 + a_7 = v_2 \right\}
    % a_7 = -a_4 + v_2
    % a_5 = v_1/(a_4a_7) = -v_1/(a_4(a_4-v_2))
    % a_4 = -a_7 + v_2
    % a_5 = -v_1/(a_7(a_7-v_2)),
\end{align}
if ideal data are observed.
\begin{comment}
 &\Leftrightarrow \left(
    \begin{array}{cc}
      y(t_1) & \left.\frac{\d y}{\d t}\right|_{t = t_1} \\
      y(t_2) & \left.\frac{\d y}{\d t}\right|_{t = t_2}
    \end{array}
  \right)\left(
    \begin{array}{c}
      d_1 \\
      d_2
    \end{array}
  \right) = -\left(
    \begin{array}{c}
      \left.\frac{\d^2 y}{\d t^2}\right|_{t = t_1} \\
      \left.\frac{\d^2 y}{\d t^2}\right|_{t = t_2}
    \end{array}
  \right).\nonumber \\
\end{comment}
%where time $t_1, t_2$ can be chosen appropriately in the way that
%the above equation
%the equation \eqref{eq:dd}
%has a unique solution.
\begin{comment}
\begin{align}
\left|{y(t_1)\left.\frac{\d y}{\d t}\right|_{t = t_2}-y(t_2)\left.\frac{\d y}{\d t}\right|_{t = t_1}}\right| \neq 0.
\end{align}
\end{comment}

%In \cite{science}, a single set of parameters for each subject is estimated as shown in \ref{tab:neumann_parameter}. On the other hand, the proposed method estimates the parameter variety of \eqref{eq:virous} given time-series data.
Next, considering the measurement noise and the lack of derivatives for the output data, %, which are not observed in \cite{science}, %are required in our method, %they , and hence,
%they need to be approximated.
data pre-processing is performed. We use the parameters estimated in \cite{science}, shown in Table \ref{tab:neumann_parameter}, for generating pseudo-data. In this case, it is possible to compute the exact solution of $y$ for \eqref{eq:exact} given the parameters in Table \ref{tab:neumann_parameter}; hence, $\frac{\d y}{\d t}, \frac{\d^2 y}{\d t^2}$ at arbitrary time points in $[T_0, T_1]$ can also be computed analytically. In order to compute $v_1, v_2$, two time points $(t_1, t_2) = (1.8594, 6.1602)$ are randomly selected. The generated pseudo-data are presented in Table \ref{tab:estimated_v}.

\begin{table}[htb]
   \caption{Generated pseudo-data, given the parameters in Table \ref{tab:neumann_parameter}. For randomly selected time points $(t_1,t_2) =(1.8594, 6.1602)$, $y,\dot{y},\ddot{y}$ are computed using \eqref{eq:exact}.}
\begin{center}
  \begin{tabular}{|c|c|c|c|c|c|c|c|c|c|c|} \hline
Subject &  $(y(t_1), \dot{y}{(t_1)}, \ddot{y}(t_1))$ & $(y(t_2), \dot{y}{(t_2)}, \ddot{y}(t_2))$  \\%&$E$\\ \hline \hline %& $t_{\rm{c}}$ & $k_{\rm{a}}$ & $k_{\rm{e}}$ \\
\hline \hline
%1-H &0(0.01) &0.75(0.1)  &6.9(0.2)\\
%2-D& 0.16(0.04) &0.95(0.1)  &5.6(0.5)\\
%3-D&0.4(0.05) &0.99(0.2) &6(0.8) \\ \hline
1-H &$(1.0251, -0.0010,0.0070)\times 10^6$ &$(1.0250,-0.0000, 0.0000)\times 10^6$  \\
2-D& $(4.1781,-0.7127,0.5485)\times 10^4$ &$(2.1676,-0.3290,0.0499)\times 10^4$  \\
3-D&$(2.4049,-1.0615,1.0792)\times 10^3$ &$(434.9356,-172.1117,68.1076)$ \\ \hline
\end{tabular} \label{tab:estimated_v}
   \end{center}
\end{table}
\begin{comment}
We applied the central difference for the numerical approximation of these values at randomly selected $t_1, t_2$ and obtain, for example, (b)
$
y(t_1) \simeq 17.7306 \times 10^3,
%\left.
\frac{\d y}{\d t}(t_1) %\right|_{t = t_1}
\simeq -2.6911 \times 10^3,
%\left.
\frac{\d^2 y}{\d t^2}(t_1)
%\right|_{t = t_1}
\simeq 69.1216$, %\mathrm{e}3 \\
$y(t_2) \simeq 39.6348 \times 10^3, %\mathrm{e}3 &
%\left.
\frac{\d y}{\d t}(t_2)
%\right|_{t = t_2}
\simeq -6.1401 \times 10^3%\mathrm{e}3
$,
$%\left.
\frac{\d^2 y}{\d t^2}(t_2)
%\right|_{t = t_2}
\simeq 4.1728 \times 10^3$.
%       4.1728 \times 10^3 %\mathrm{e}3
\end{comment}

Substituting the pseudo-data in \eqref{eq:find_v}, we obtain $v_1, v_2$ for $1-H, 2-D, 3-D$ as follows, in order:
 \begin{align}
% \hspace{-1zh}
 \left(
    \begin{array}{c}
      v_1 \\
      v_2
    \end{array}
  \right)
%  & = -{\left(
%    \begin{array}{cc}
%      17.7306 \times 10^3 & %\mathrm{e}3 &
%      -2.6911 \times 10^3 \\%\mathrm{e}3 \\
%%     39.6348 \times 10^3 & %\mathrm{e}3 &
%     -6.1401 \times 10^3 %\mathrm{e}3
%    \end{array}
%  \right)}^{-1} \notag \\
%  & \qquad \left(
%%    \begin{array}{c}
%       69.1216\\
%       4.1728 \times 10^3 %\mathrm{e}3
%    \end{array}
%  \right)\nonumber \\
 &\simeq \left(
    \begin{array}{cc}
    0.0000 \\
   6.9000
    \end{array}
  \right),
  \left(
    \begin{array}{cc}
    0.8512 \\
   5.7600
    \end{array}
  \right),
  \left(
    \begin{array}{cc}
    2.3760 \\
   6.4000
    \end{array}
  \right)
  . \label{eq:2dd}
 \end{align}
Substituting the above in \eqref{eq:ioeq}, the algebraic varieties containing the parameter variety of \eqref{eq:math_eq}, \eqref{eq:out_eq}, given the pseudo-data shown in Table \ref{tab:estimated_v}, are obtained. The variety for $2-D, 3-D$ is visualized in Fig. \ref{fig:pv_2_3}.

 \begin{figure}[t]
\centering
\includegraphics[scale=0.25]{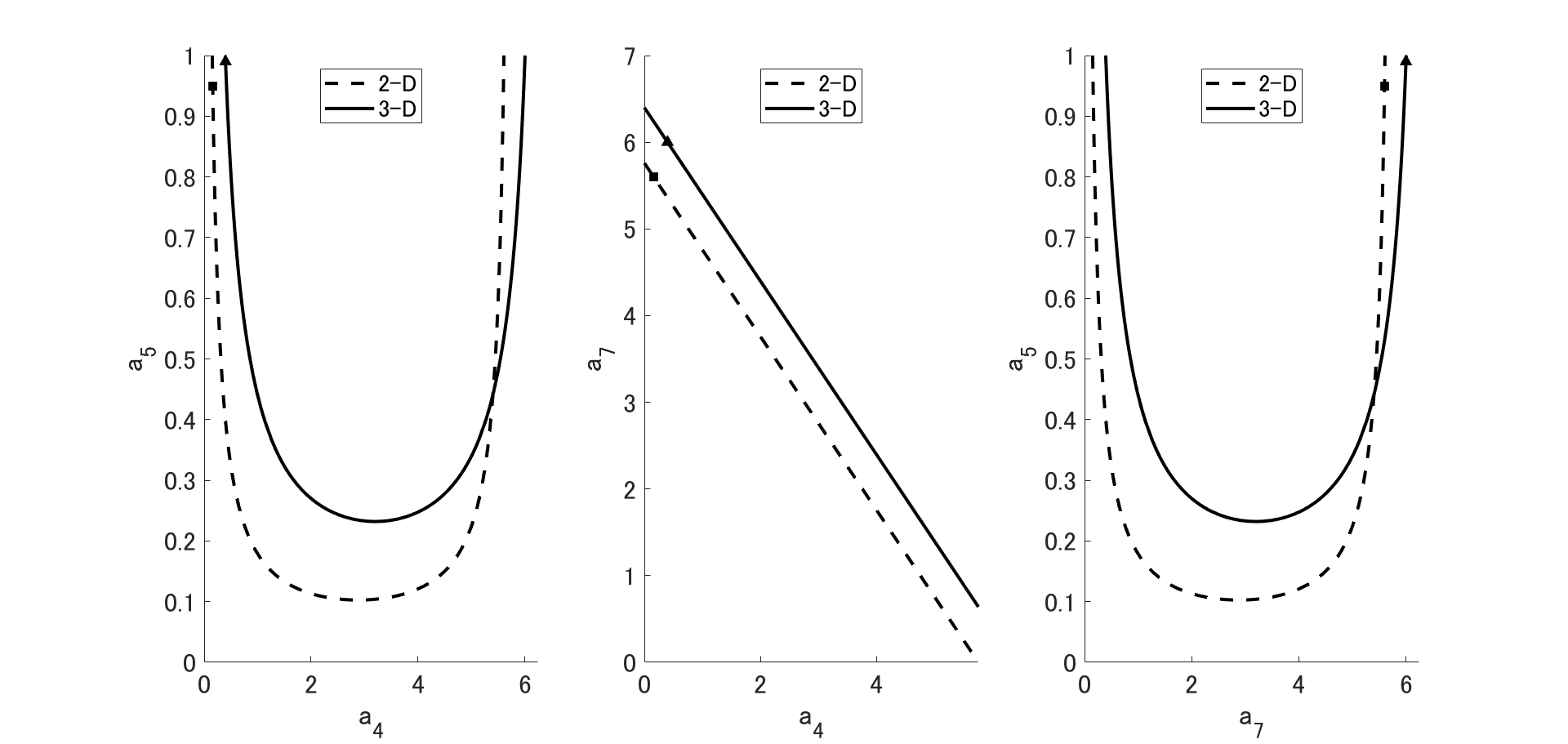}
\caption{Algebraic variety containing the parameter variety of \eqref{eq:virous}, given the pseudo-data for $2-D$ and $3-D$. The symbols on the variety are the estimated parameters shown in Table \ref{tab:neumann_parameter}. The varieties are projected into two- dimensional parameter spaces: $(a_4,a_5), (a_4,a_7),(a_7,a_5)$. Following the assumption on $a_5$, the varieties are plotted, where $0 \leq a_5 \leq 1$.\label{fig:pv_2_3}}
\end{figure}

%\begin{comment}
\subsection{Analysis of the varieties}%\label{sec:app}
There are several findings on the model and the given time-series data based on the estimated varieties.

First, the range of some of the parameters are revealed. The constraint equation \eqref{eq:pv_virus} can be rearranged as
 \begin{align*}
a_5 = -\frac{v_1}{a_4(a_4-v_2)}, a_7 = -a_4 + v_2,
 \end{align*}
if $a_4 \neq 0, a_4 \neq v_2$. For $2-D, 3-D$, it is apparent that $a_4 > 0$ from the first constraint of \eqref{eq:pv_virus}. As $a_5 \geq 0$ according to the model assumption and $v_1 > 0$ from the estimated results, $a_4 \leq v_2$ should hold. Thus, the possible ranges of $a_4$ and $a_7$ are revealed through the estimation of the varieties as $0 < a_4, a_7 < v_2$.

Next, the constraints defining the variety for $1-H$ are explicitly described as
\begin{align}
a_4a_5a_7= 0, a_4 + a_7 = 6.9. %\mathrm{e}-6.
\label{eq:pv_1}
\end{align}
Thus, if $a_4 = 0$ or $a_7 = 0$, we obtain constraint $a_7 = 6.9$ or $a_4 = 6.9$, respectively, and $a_5$ is not constrained. On the other hand, if $a_5 = 0$, $a_4 + a_7 = 6.9$ should hold. This implies that the variety for $1-H$ is structurally different from those of both $2-D$ and $3-D$. According to \cite{science}, the dynamics of the fitted curve for $1-H$ is significantly different from those for $2-D$ and $3-D$ as shown in Fig. \ref{fig:data}, which is supposed to reflect the difference in the assigned regiments as mentioned in \cite{science}; this corresponds to our result. Hence, by checking the geometry of the varieties determined from the data, the abnormal dynamics of the model may be captured.

Finally, we revisit the results obtained in \cite{science}. $a_4 = 0$ implies that the infected cells live longer than at least the measured duration. On the other hand, $a_5 = 0$ indicates that IFN does not work at all for the subjects. Therefore, the data can be explained not only with the extremely low death rate of the infected cells as mentioned in \cite{science}, but also with the invalidity of IFN, even though the parameters estimated in the original study are correctly located in the estimated varieties. Technically, in the original study, $a_5$ and $a_4$ are estimated in order, resulting in a possible value of zero for $a_4$. On the other hand, our method successfully extracted a new insight on the dynamics of the system, which was hidden in the data and overlooked by conventional approaches.

\section{Theoretical aspects of the proposed method}\label{sec:theory}
\subsection{Mathematical formulations}\label{sec:formulation}
In order to investigate model \eqref{eq:math_eq}, \eqref{eq:out_eq}, given observed time-series data \eqref{eq:in_data}, \eqref{eq:out_data}, the input-output relationships of the model are first examined. The investigation of the input-output relationships involves differential algebraic manipulations, as performed in Section \ref{sec:step1}, and is equivalent to the investigation of the differential ideal generated by the model.
\begin{definition}
Let $\mathcal{R}$ be a ring.
The derivation is a map $\partial: \mathcal{R} \to \mathcal{R}$ such that
\begin{itemize}
\item $\partial (r+s) = \partial(r) + \partial(s)$ for all $r, s \in \mathcal{R}$
\item and $\partial (r \cdot s) = r \cdot \partial(s) + \partial(r) \cdot s$ for all $r, s \in \mathcal{R}$.
\end{itemize}
\end{definition}
\begin{definition}
Differential ring $\mathcal{R}$ is a commutative ring with derivatives $\partial:\mathcal{R} \rightarrow \mathcal{R}$. % and $\mathcal{R}$ is closed under the derivations.
\end{definition}
\begin{definition}\label{def:differentialideal}
The differential ideal $I$ of differential ring $\mathcal{R}$ equipped with derivation operators $\partial$ is an ideal that is closed under the derivations.
\end{definition}

The differential ideal generated by model \eqref{eq:math_eq}, \eqref{eq:out_eq} is defined as follows:
\begin{definition}
Let $\mathbb{C}(a)$ be a field of rational functions of parameter vector $a$ with complex coefficients. Given state-space model \eqref{eq:math_eq}, \eqref{eq:out_eq}, where
\begin{align*}
x = \left(x_1, \ldots, x_N\right)^T, f(x,u;a) = \left(f_1(x,u;a),\ldots,\, f_N(x,u;a) \right)^T,
\end{align*}
 the state-space differential ideal of model \eqref{eq:math_eq}, \eqref{eq:out_eq} % of system
 is defined by the differential ideal generated by the model, i.e.,
\begin{align*}
\Sigma := [\dot{x}_1 -f_1(x,u;a), \ldots, \dot{x}_N -f_N(x,u;a), y-g(x,u;a)] \subset \mathbb{C}(a)\left\{x_1, \ldots, x_N, u, y\right\}.
\end{align*}
\end{definition}
From an algebraic point of view, the input-output relationships to be investigated %through differential algebraic manipulations
are in a specific subset of the above ideal:
\begin{definition}
The input-output differential ideal of model \eqref{eq:math_eq}, \eqref{eq:out_eq}
is defined by the following sub-ideal of state-space differential ideal $\Sigma$,
\begin{align}\label{eq:iodiffideal}
\Sigma_{\mathrm{IO}} := \Sigma \cap  \mathbb{C}(a)\left\{u,y\right\}.%\subset \mathcal{Q}\left<a\right>\left\{x_1, \ldots, x_n\right\}.
\end{align}
\end{definition}

Under the framework of differential algebra \cite{ritt, kolchin}, the input-output relationships of \eqref{eq:math_eq}, \eqref{eq:out_eq} are investigated, for example, in structural identifiability analysis \cite{fliess, ljung1994, Saccomani_2003}. However, as mentioned in Remark \ref{remark:differentialalgebra}, infinite derivative operations of the model may not be plausible in practice. Therefore, instead of following a differential algebraic framework, we restrict ourselves to a non-differential (i.e., standard) algebraic framework; i.e., instead of considering $\Sigma, \Sigma_{\mathrm{IO}}$, we consider the non-differential ideals obtained by truncating them such that higher derivatives of the variables more than necessary do not appear. Through such truncation, the input-output relationships of \eqref{eq:math_eq}, \eqref{eq:out_eq} can be investigated without considering the higher derivatives of the variables than required. As mentioned in Section \ref{sec:step1}, there exists an upper bound on the number of differentiations required to obtain at least one element in $\Sigma_{\mathrm{IO}}$. This is guaranteed by the following proposition \cite{forsman, jirstrand, meshkatn}:
\begin{prop}
For model \eqref{eq:math_eq}, \eqref{eq:out_eq}, if ideal $\IN$ is defined as
		\begin{align*}
		 \IN &=
			\langle \frac{\d x}{\d t} - f,\ldots,\frac{\d^N x}{\d t^N} - f^{(N-1)}, \ldots y-g, \ldots, y^{(N)}-g^{(N)} \rangle
		\\
		& \subset \mathbb{C}(a)[x,\ldots, x^{(N)}, y, \ldots, y^{(N)}, u, \ldots, u^{(N)}],
	\end{align*}
	$\IN \cap \mathbb{C}(a)[y, \ldots, y^{(N)}, u, \ldots, u^{(N)}]$ is not a zero ideal.
\end{prop}
Note that $\IN$ is the ideal obtained by truncating $\Sigma$ up to the $N$-th order derivatives of $x,y$ and $u$.
Through Algorithm \ref{algorithm:alg2}, the minimal number of differentiations required to obtain equations that do not contain state variables nor their derivatives for the model is computed \cite{jirstrand}.
Denoting this integer as $L \leq N$,
%\begin{definition}
%is defined as follows :
with a sufficiently differentiated \eqref{eq:math_eq}, \eqref{eq:out_eq},
%introduce the following definitions :
\begin{align}
    %I_l := \langle &
    \begin{aligned}
    &x^{(1)} - f(x,u; a) ,\ldots, x^{(L)} - \frac{\d^{(L-1)}}{\d t^{(L-1)}}f(x,u;a),\\
    &y-g(x,u;a), \ldots, y^{(L)}-\frac{\d^{(L)}}{\d t^{(L)}}g(x,u;a) \label{eq:diff_model}
    %\rangle
    \end{aligned},
\end{align}
where $L$ is a minimal integer that is an output of Algorithm \ref{algorithm:alg2}, by which an ideal generated by \eqref{eq:diff_model} contains at least one element. The ideal generated by \eqref{eq:diff_model} is denoted as follows:
\begin{align}
\begin{aligned}
    \Il := \langle &x^{(1)} - f(x,u; a) ,\ldots, x^{(L)} - \frac{\d^{(L-1)}}{\d t^{(L-1)}}f(x,u;a), \\
    &y-g(x,u;a), \ldots, y^{(L)}-\frac{\d^{(L)}}{\d t^{(L)}}g(x,u;a) \rangle \\
    & \subset \mathbb{C}(a)[x,\ldots, x^{(L)}, y, \ldots, y^{(L)}, u, \ldots, u^{(L)}]. \label{eq:diff_model_ideal}
    \end{aligned}
\end{align}%\leq l \leq N$
%\end{definition}

\begin{remark}\label{remark:smoothness}
This implies that the assumption on the smoothness of the solutions of \eqref{eq:math_eq}, \eqref{eq:out_eq} is relaxed to $C^L$ functions for $x,y$ and $C^{L-1}$ functions for $u$.
\end{remark}

Based on the above, we provide the algebraic counterpart of Definition \ref{dfn:N-pv}:
\begin{definition}\label{dfn:N-pv2}
Suppose that time-series data \eqref{eq:in_data}, \eqref{eq:out_data} corresponding to the input and output of state-space model \eqref{eq:math_eq}, \eqref{eq:out_eq} are observed.
We define the parameter variety of the model, given data $\hu$ and $\hy$, by the set of parameters with which %the model gives $u$ and $y$ coincident with the data.
\begin{align*}
    &\forall t_i \in [T_0, T_1], \exists \left(x(t_i), \ldots, x^{(L)}(t_i), %y_\d(i,0),\ldots, y_\d(i,j), {u_\d}_1(i,0),\ldots,{u_\d}_M(i,j)
    y(t_i), \ldots, y^{(L)}(t_i), u(t_i), \ldots, u^{(L)}(t_i)
    \right)^T \in \mathbb{C}^{(N+1+M) \times L}
\end{align*}
such that \eqref{eq:diff_model} and \eqref{eq:model_data} hold, where $i = 1,\ldots, L$ and $L$ is the minimal number of differentiations required.
%there exist $\left(x, \ldots, x^{(L)},y,\ldots, y^{(L)}, u,\ldots,u^{(L)}\right)^T \in \mathbb{C}^{(N\times L)}$ on $V(I_l)$ that satisfies
\end{definition}

\begin{definition}
The input-output equations of \eqref{eq:math_eq}, \eqref{eq:out_eq} are defined as equations corresponding to the polynomials in the ideal defined as follows:
\begin{align*}
    \Il \cap \mathbb{C}(a)[y,\ldots, y^{(L)}, u, \ldots, u^{(L)}].
\end{align*}
\end{definition}

\begin{comment}
\begin{remark}
The Gr\"obner basis of %an ideal that is formed when the input-output equation is appeared in the above method
the ideal generated by \eqref{eq:diff_model} with respect to lexicographic ordering ${x}^{(L)} > \cdots> x > {y}^{(L)} > \cdots > y > u^{(L)} > \cdots > u$ is utilized to remove the parameters that are not included in the parameter variety of the output of Algorithm \ref{algorithm:alg1}. If the user does need not consider such parameters, a method based on the Rosenfeld-Gr\"obner basis may suffice. See \textcolor{red}{???} for details.
\end{remark}
\end{comment}

\begin{comment}
Substituting the second equation and its derivative in the first one, we get
\begin{align*}
    \frac{\d y}{\d t} = a_1y\left(1-\frac{y}{a_2}-a_3\frac{u}{a_2}\right),
\end{align*}
which does not contain $x_1$. For such simple models, it appears easy to obtain equations that represent the input-output relationships; however, %if it's not the case,
manipulations to eliminate the state variables become complicated depending on the models, and different equations may be obtained depending on the manipulations.
\end{comment}

\subsection{Main theorems} \label{sec:proof}\label{sec:main}
Before presenting the main theorems, certain definitions and a proposition are introduced.
\begin{definition}
For ideal $I$ on polynomial ring $K[x_1, \ldots, x_n]$ over field $K$ with a given monomial order, if its generator $G = \left\{ g_1,\ldots, g_s \right\}$ satisfies
 \begin{center}
 $f \in I \Leftrightarrow f \mathrm{\ is \ divisible \ by\ } G,$,
 \end{center}
 $G$ is called the Gr\"obner basis of $I$.
\end{definition}

\begin{definition}
A set of monomials $\left\{u_1,\ldots, u_t\right\}$ appears in a non-zero polynomial $f = a_1u_1 +\cdots + a_tu_t$ of $K[x]$, %=K[x_1,x_2,\ldots, x_n]$,
where $0 \neq a_i \in K (i = 1,\ldots, t)$ is called the support. The largest monomial in the support of polynomial $f$ with respect to a given monomial order $<$ of ring $K[x]$ is denoted by ${\rm{in}}_<(f)$.
The Gr\"obner basis $\left\{ g_1, \ldots, g_s\right\}$ of ideal $I$ on polynomial ring $K[x]$ with respect to monomial order $<$ is the reduced Gr\"obner basis, if the followings are satisfied:
\begin{itemize}
\item The coefficient of ${\rm{in}}_<(g_i)$ of polynomial $g_i$ is 1 ($1 \leq i \leq s$)
\item and if $i \neq j$, no monomial in the support of $g_j$ is divisible by ${\rm{in}}_<(g_i)$.
\end{itemize}
\end{definition}
The reduced Gr\"obner basis of ideal $I$ on polynomial ring $K[x]$ with respect to a given monomial order is unique \cite{Cox}. The basis of the input-output equations of \eqref{eq:math_eq}, \eqref{eq:out_eq}, represented as \eqref{eq:ioeq}, is the reduced Gr\"obner basis of $\Il$ with respect to lexicographic order $x^{(i)} > \cdots > x > y^{(i)} > \cdots > y > u^{(i)} > \cdots > u$. The following is called the elimination theorem \cite[p. 122]{Cox}.
\begin{prop}\label{prop:elim}
Suppose that lexicographic ordering is used as the monomial ordering in a ring of polynomials $K[x_1, \ldots, x_n]$ over field $K$ such that $x_1 > x_2 > \ldots > x_n$.
If $G$ is reduced, the Gr\"obner basis of ideal $I \subset K[x_1,\ldots,x_n]$ for every $1 \leq j \leq n$, $G \cap K[x_{j},\ldots,x_n]$ is the Gr\"obner basis of $I \cap K[x_j, \ldots, x_n]$.
\end{prop}

Now, we prove that the algebraic variety obtained using Algorithm \ref{algorithm:alg1} contains the parameter variety of \eqref{eq:math_eq}, \eqref{eq:out_eq}, given \eqref{eq:in_data}, \eqref{eq:out_data}. In the following, $\mathrm{V}(f_1,\ldots, f_s)$ denotes the algebraic variety defined as $f_1,\ldots, f_s \in K[x_1,\ldots, x_n]$:
\begin{align*}
    \mathrm{V}(f_1,\ldots,f_s) = \left\{ (a_1,\ldots, a_n) \in K^n \mid f_i(a_1,\ldots, a_n) = 0 \mathrm{\ for\ all\ } 1 \leq i \leq s \right\}.
\end{align*}

\begin{theorem}
  %Suppose that there is no dependency between $ u, u^{(1)}, \ldots \in \Sigma$.
  The parameter variety of \eqref{eq:math_eq}, \eqref{eq:out_eq}, given \eqref{eq:in_data}, \eqref{eq:out_data}, is a subset of $\left\{ a \mid c_l(a) = v_l, l = 1,\ldots, L\right\}$ obtained through Algorithm \ref{algorithm:alg1}.
\end{theorem}

\begin{proof}
We denote $V \subset \mathbb{C}^{(N+1+M)\times L}$ as the set of common roots of \eqref{eq:diff_model}. Let $J$ be an ideal such that $J = \Il \cap \mathbb{C}(a)[y,\ldots, y^{(L)},u,\ldots, u^{(L)}]$. Suppose that the projection map
\begin{align*}
    \pi_{y,\ldots, y^{(L)},u,\ldots, u^{(L)}}:\mathbb{C}^{(N+1+M)\times L} \rightarrow \mathbb{C}^{(1+M)\times L} %(V) \subseteq V
\end{align*}
sends $\left(x, \ldots, x^{(L)},
    y, \ldots, y^{(L)}, u, \ldots, u^{(L)}\right)^T$ to $\left(y, \ldots, y^{(L)}, u, \ldots, u^{(L)}\right)^T$.
According to p. 129 in \cite{Cox}%\cite[p. 129]{Cox}
, the following holds:
\begin{align}\label{eq:closure_thm}
   \pi_{y,\ldots, y^{(L)},u,\ldots, u^{(L)}}(V) \subseteq \mathrm{V}(J). %= \mathrm{V}(C_\mathrm{IO})
\end{align}
Let $h$ be the basis of the input-output equations of \eqref{eq:math_eq}, \eqref{eq:out_eq}, i.e., \eqref{eq:ioeq}, which is the output of Algorithm \ref{algorithm:alg2}. According to the algorithm, $h$ is obtained as an element of the reduced Gr\"obner basis of the sufficiently differentiated \eqref{eq:math_eq}, \eqref{eq:out_eq}, $\Il$, with respect to lexicographic ordering $x^{(L)}>\ldots> x> y^{(L)}> \ldots >y> u^{(L)}>\ldots> u$. Hence, based on
the elimination theorem (Proposition \ref{prop:elim}), it follows that
\begin{align*}
    \mathrm{V}(J)= \mathrm{V}(\langle h \rangle).
\end{align*}
Therefore, for all $p \in \mathrm{V}(J)$, $p \in \mathrm{V}(\langle h \rangle) = V(h)$, i.e., $p$ satisfies $h = 0$. Based on \eqref{eq:closure_thm}, the following holds:
\begin{align}\label{eq:satisfy_ioeq}
    \forall q \in \pi_{y,\ldots, y^{(L)},u,\ldots, u^{(L)}}(V)\subseteq
    \mathrm{V}(J),\quad q \,\mbox{satisfies}\, h = 0.
\end{align}

Let $\alpha \in \mathbb{C}^n$ be a parameter in the parameter variety of \eqref{eq:math_eq}, \eqref{eq:out_eq}, given \eqref{eq:in_data}, \eqref{eq:out_data}, and $V_\alpha \subset \mathbb{C}^{(N+1+M)\times L}$ is the set of common roots of \eqref{eq:diff_model} given $\alpha$.
\begin{comment}
From the definition and \eqref{eq:satisfy_ioeq}, $\alpha$ satisfies
\begin{align}
y_\m(t_1;\alpha)&=Y_\d(1,0),\ldots, y_\m(t_K;\alpha) = Y_\d(K,0),\ldots,\\
y_\m^{(N)}(t_1;\alpha) &= Y_\d(1,N), \ldots, y_\m^{(N)}(t_K;\alpha) = Y_\d(K,N) \label{eq:out_data_eq},
\end{align}
where $y_\m^{(j)}(t_k;\alpha)$ is the $j$-th derivative of output $y$ of \eqref{eq:math_eq}, \eqref{eq:out_eq} at $t = t_k$, given $a = \alpha$. %and input $u(t) = u_d(t)(\forall t \in [T_0, T_1])$ %with initial conditions $x(0)$.
Also note that $\forall t \in [T_0, T_1]$ and the following holds:
\begin{align*}
    \left(u(t),\ldots, u^{(N)}(t)\right)^T = U_\d(t). \label{eq:input_data_eq}
\end{align*}
\end{comment}
Based on Definition \ref{dfn:N-pv2}, the common roots of %following holds :
\begin{align}
    %I_l := \langle &
    \begin{aligned}
    &x^{(1)} - f(x,u; \alpha) ,\ldots, x^{(L)} - \frac{\d^{(L-1)}}{\d t^{(L-1)}}f(x,u;\alpha),\\
    &y-g(x,u;\alpha), \ldots, y^{(L)}-\frac{\d^{(L)}}{\d t^{(L)}}g(x,u;\alpha) \label{eq:diff_model_alpha}
    %\rangle
    \end{aligned}
\end{align}
such that \eqref{eq:model_data} holds are in $\pi_{y,\ldots, y^{(L)},u,\ldots, u^{(L)}}(V_\alpha)$. Therefore, for each $t_i \in [T_0, T_1] (i =1,\ldots, K)$, with a fixed $\alpha$, the following holds :
\begin{align*}
    &\Sigma_{l = 1}^L c_l(\alpha)f_l(\hu(t_i),\hy(t_i),\dot{\hu}(t_i), \dot{\hy}(t_i),\ldots, {\hu}^{(N)}(t_i), \hy^{(N)}(t_i)) \\
    &= f_{L+1}(\hu(t_i),\hy(t_i),\dot{\hu}(t_i), \dot{\hy}(t_i),\ldots, {\hu}^{(N)}(t_i), \hy^{(N)}(t_i)).
\end{align*}
This leads to simultaneous linear equations
\begin{align}
\label{eq:sys_ioeq_alpha}
&\left(
\begin{array}{ccc}
f_1(\hu(t_1), \hy(1,:))& \cdots& f_{L}(\hu(t_L), \hy(L,:))\\
\vdots& \ddots & \vdots\\
f_1(\hu(t_1), \hy(L,:))& \cdots & f_{L}(\hu(t_1), \hy(L,:))\\
\end{array}
\right) \left(
\begin{array}{c}
c_1(\alpha)\\
\vdots \\
c_{L}(\alpha)\\
\end{array}
\right) %\nonumber \\
%&
= \left(\begin{array}{c}
f_{K+1}(\hu(t_1), \hy(1,:))\\
\vdots \\
f_{K+1}(\hu(t_K), \hy(K,:))\\
\end{array}\right).
\end{align}
As the matrix on the left-side and the vector on the right-side of \eqref{eq:sys_ioeq_alpha} are independent of $\alpha$, %\eqref{eq:sys_ioeq_alpha} holds for all $\alpha$ in the parameter variety and
%a unique solution of
%Regarding
\eqref{eq:sys_ioeq_alpha} as a simultaneous linear equation of %, which is formed by $h$ where \eqref{eq:model_data} is substituted,
%regarding it as the equations of
${\left(c_1(\alpha), \ldots, c_L(\alpha) \right)}^T$ admits a unique solution, which is independent of $\alpha$. We denote this as $(v_1,\ldots, v_L)^T$. Thus, we obtain
\begin{align*}
     c_l(\alpha) = v_l, l = 1,\ldots, L.
\end{align*}
This concludes that $\alpha$ is in $\left\{ a \mid c_l(a) = v_l, l = 1,\ldots, L\right\}$.

\end{proof}

Next, the sufficient condition for the algebraic variety, obtained using Algorithm \ref{algorithm:alg1}, to be the parameter variety is presented by extending the so-called extension theorem \cite[p. 125]{Cox}: % into our framework, providing a method to get rid of parameters that are in the parameter variety but not in the algebraic variety derived through \ref{algorithm:alg1}.

\begin{comment}
Based on the extension theorem \cite[p. 169]{cox}, differential polynomials $g_1, \ldots, g_s \in \mathbb{C}(a)[y, \dot{y}, \ldots,y^{(i)},u,\dot{u},\ldots, u^{(i)},x_1,\ldots, x_N^{(i)}]$ can be computed such that
    \begin{align}\label{eq:extension}
    V(J) = \pi_{y,\ldots, y^{(L)}, u,\ldots, u^{(L)}}(V(\Il)) \cup  W_l%\left( V(J) \cap V(g_1, \ldots, g_l)\right),
\end{align}
where
\begin{align}\label{eq:singular}
W_l := V\left(%\Sigma \cap \mathbb{C}(a)[y,\ldots,y^{(L)}, u,\ldots, u^{(L)}]
J\right)
\cap V(g_1, \ldots, g_s).
\end{align}
This leads to the following theorem:
\end{comment}
\begin{prop}\label{prop:extension}
Let $I = \langle F_1,\ldots, F_s \rangle \subseteq \mathbb{C}[x_1,\ldots, x_n]$, and let $I_1$ be the first elimination ideal of $I$. For each $1 \leq i \leq s$, $F_i$ is expressed in the form
\begin{align*}
    F_i = p_i(x_2,\ldots, x_n)x_1^{D_i} + \mbox{(terms\, in\, which\,}\ x_1\ \mbox{has\, degree\,} < D_i),
\end{align*}
where $D_i \geq 0$ and $p_i \in \mathbb{C}[x_2,\ldots, x_n]$ are non-zero. Suppose that we have a partial solution $(\tilde{x}_2,\ldots, \tilde{x}_n) \in \mathrm{V}(I_l)$. If $(\tilde{x}_2,\ldots, \tilde{x}_n) \notin \mathrm{V}(p_1,\ldots, p_s)$, there exists $\tilde{x}_1 \in \mathbb{C}$ such that $(\tilde{x}_1,\tilde{x}_2,\ldots, \tilde{x}_n) \in \mathrm{V}(I)$.
\end{prop}

The above can be used when eliminating multiple variables by confirming the sufficient condition for each eliminated variable step-by-step. More precisely, suppose $x_1,\ldots, x_m$ are the variables to be eliminated, where $m < n$; let $I_{m}$ be the $m$-th elimination ideal of $I$. Suppose that we also have a partial solution $(\tilde{x}_{m+1},\ldots, \tilde{x}_n) \in \mathrm{V}(I_m)$. To state that there exists $(\tilde{x}_1,\ldots, \tilde{x}_m) \in \mathbb{C}^m$ such that $(\tilde{x}_1,\ldots,\tilde{x}_n) \in \mathrm{V}(I)$, it is sufficient to check whether for each $j, j = m, m-1,\ldots, 1$, partial solution $(\tilde{x}_{j+1}, \ldots,\tilde{x}_n) \in \mathrm{V}(I_j)$ can be extended to $(\tilde{x}_{j}, \tilde{x}_{j+1}, \ldots,\tilde{x}_n) \in \mathbb{C}^{n-j+1}$, i.e., there exists $\tilde{x}_j \in \mathbb{C}$ such that $(\tilde{x}_{j}, \tilde{x}_{j+1}, \ldots,\tilde{x}_n) \in \mathrm{V}(I_{j-1})$, where $I_j$ is the $j$-th elimination ideal of $I$. See \cite{Cox} for the details.

Now, by incorporating the above, we provide the following theorem indicating a method to remove the parameters that are in the parameter variety but not in the algebraic variety derived using Algorithm \ref{algorithm:alg1}.
\begin{theorem}\label{thm:extention}
Let $\left\{G_1, \ldots, G_s = h\right\}$ be the reduced Gr\"{o}bner basis of
\begin{align*}
\Il \subset \mathbb{C}(a)[x,\ldots, x^{(L)}, y, \ldots, y^{(L)}, u, \ldots, u^{(L)}]
\end{align*} with respect to lexicographic ordering
$%\begin{align*}
{x}_N^{(L)} >\cdots x_1^{(L)} > \cdots >\dot{x}_N \cdots > \dot{x}_1 > x_N > \cdots > x_1 > {y}^{(L)} > %\\
\cdots > y > u^{(L)} > \cdots > u
$. %\end{align*}
%The following variables are introduced:
%\begin{align*}
$z_j$ denote variables
\begin{comment}
$
    z_0 := y^{(L)}, z_1 := x_1, z_2 := x_2, \ldots, z_N := x_N, z_{N+1} := \dot{x}_1, \ldots, z_{N\times(L+1)} := x^{(L)}_N,
$
\end{comment}
$
   z_{N\times(L+1)} := x_1, z_{N\times(L+1)-1} := x_2, \ldots, z_{N\times(L+1)-(N-1)} := x_N, z_{N\times(L+1)-(N-2)} := \dot{x}_1, \ldots, z_{1} := x^{(L)}_N,
$
%\end{align*}
and $E_j$ denotes
a set of variables
\begin{comment}
$    %E_j =
       \left\{ u, \ldots, u^{(L-1)}, y, \ldots, y^{(L)} = z_0, z_{1},\ldots,z_{j-1}\right\}. %\quad 2 \leq j \leq N\times(L+1)+1
       $.
\end{comment}
\begin{align*}    %E_j =
       \left\{ u, \ldots, u^{(L-1)}, y, \ldots, y^{(L)}, z_{N(L+1)},\ldots,z_{j+1}\right\}. %\quad 2 \leq j \leq N\times(L+1)+1.
\end{align*}
%$[x,\ldots, x^{(L)}, y, \ldots, y^{(L)}, u, \ldots, u^{(L-1)}] \backslash z_{j},\ldots, z_{N\times(L+1)}$
For each $1 \leq i \leq s$, express $G_i$ in the form{\normalfont
\begin{align*}
    G_i = p_i^j(E_j)z_j^{D_i^j} + \mbox{(terms\, in\, which\,}\ z_j\  \mbox{has\, degree\,} < D_i^j),
\end{align*}
}\noindent
where $p_i^j(E_j) \in \mathbb{C}(a)[E_j]$.
%Suppose that we have a partial solution
%\begin{align*}
%\left(z_{j-1},\ldots,z_1, y, \ldots, y^{(L)}, u, \ldots, u^{(L)}\right) \in \mathrm{V}(I_{[N\times(L+1) - (j-1)]})
%\end{align*}
%where
\begin{comment}
Let %$I_{[N\times(L+1) - (j-1)]}$
$I_j$ be the %$\left( N\times(L+1) - (j-1) \right)$
$j$-th elimination ideal of $\Il$:
\begin{align*}
    I_{j} = \Il \cap \mathbb{C}(a)[z_{j+1},\ldots,z_{N(L+1)}, y, \ldots, y^{(L)}, u, \ldots, u^{(L-1)}].
\end{align*}
\end{comment}
For each $1 \leq j \leq N\times(L+1)$, we express
\begin{align}\label{eq:singular}
    %W_j :=\mathrm{V}(I_j) \cap \mathrm{V}(P_j)
    W_j :=\mathrm{V}(\langle G_1,\ldots, G_s\rangle \cap \mathbb{C}(a)[E_{j}]) \cap \mathrm{V}(P_j)
\end{align}
in %$\mathbb{C}^{j-1 + (1+M)\times (1+l)}$
$\mathbb{C}^{(N+1+M)\times (L+1)-j}$,
where $P_j \subset \mathbb{C}(a)[E_j]$ denotes the non-empty set of polynomials containing $p_i^j(1,\ldots, s)$ in which $\left\{G_1,\ldots, G_s\right\} \cap \mathbb{C}(a)[E_{j-1}]$ such that $D_i^j \geq 0$ and $p_i^j(E_j) \neq 0$.
If
\begin{align}\label{eq:suffcond}
%\left\{ W_j \mid j = 1,\ldots, N\times(L+1) \right\} = \phi,
 W_j  = \phi \quad \mbox{for all\ } j = 1,\ldots, N\times(L+1),
\end{align}
\begin{align*}
%\forall &\alpha \in
&\left\{ a \mid c_l(a) = v_l, l = 1,\ldots, L\right\}% \\
%&
\subset \mbox{the\, parameter\, variety\, of\, \eqref{eq:math_eq}, \eqref{eq:out_eq} given \eqref{eq:in_data}, \eqref{eq:out_data}}.
\end{align*}
\end{theorem}

\begin{proof}
Let $\beta \in \mathbb{C}^n$ be an element of $\left\{ a \mid c_l(a) = v_l, l = 1,\ldots, L\right\}$. As $c_l(\beta) = v_l (l = 1,\ldots, L)$, \eqref{eq:sys_ioeq_alpha} holds on replacing $\alpha$ with $\beta$. This implies that \begin{align}\label{eq:partial_io}
&(\hu(t_i),\hy(t_i),\ldots, {\hu}^{(N)}(t_i), \hy^{(N)}(t_i)) \nonumber \\
=& (\tu(t_i),\ty(t_i; \beta),\ldots, {\tu}^{(N)}(t_i), \ty^{(N)}(t_i; \beta)) \in \mathrm{V}(J_\beta)\, (i = 1,\ldots, L),
\end{align}
where $\tu(t_i)$ and $\ty(t_i; \beta)$ denote the inputs and output of \eqref{eq:math_eq}, \eqref{eq:out_eq}, respectively, and $J_\beta$ denotes the elimination ideal $J = \Il \cap \mathbb{C}(a)[y,\ldots, y^{(L)},u,\ldots, u^{(L)}]$, given a fixed parameter $\beta$. Thus, $(\tu(t_i),\ty(t_i; \beta),\ldots,$
${\tu}^{(N)}(t_i), \ty^{(N)}(t_i; \beta))$ is shown to be a partial solution of $\mathrm{V}(J_\beta)$. % where $J_\beta$ is the elimination ideal $J$ given a fixed parameter $\beta$
Therefore, in order to show that $\beta$ is in the parameter variety of \eqref{eq:math_eq}, \eqref{eq:out_eq}, given \eqref{eq:in_data}, \eqref{eq:out_data}, it remains to prove that there exists $(\tilde{x}(t_i;\beta), \ldots, \tilde{x}^{(L)}(t_i;\beta)) \in \mathbb{C}^{N\times(L+1)}\, (i = 1,\ldots, K)$ such that the following holds:
\begin{align}\label{eq:toshow}
(\tu(t_i),\ty(t_i; \beta),\ldots, {\tu}^{(N)}(t_i), \ty^{(N)}(t_i; \beta),\tilde{x}(t_i; \beta),\ldots,\tilde{x}^{(L)}(t_i; \beta))
\in \mathrm{V}(I^{(L)}).
\end{align}

To show this, according to Proposition \ref{prop:elim}, it is sufficient that for each $j (j = N\times(L+1),\ldots, 1)$, the partial solution
%$(\tilde{u}, \ldots, \tilde{u}^{(L-1)}, \tilde{y}, \ldots, \tilde{y}^{(L)}, \tilde{z}_{N(L+1)},\ldots,\tilde{z}_{j+1}) \in \mathrm{V}(I_j)$
\begin{align}\label{eq:partial1}
(\tu(t_i),\ty(t_i; \beta),\ldots, {\tu}^{(N)}(t_i), \ty^{(N)}(t_i; \beta),\tilde{z}_{N(L+1)}(t_i;\beta),\ldots,\tilde{z}_{j+1}(t_i;\beta))
\end{align}
can be extended to
\begin{align}\label{eq:partial2}
&(\tu(t_i),\ty(t_i; \beta),\ldots, {\tu}^{(N)}(t_i), \ty^{(N)}(t_i; \beta), \tilde{z}_{N(L+1)}(t_i;\beta),\ldots,\tilde{z}_{j+1}(t_i;\beta), \tilde{z}_j(t_i;\beta) ) %\nonumber \\
\in%&\, 
\mathbb{C}^{(N+1+M)\times(L+1)-j+1},
\end{align}
where $\tilde{z}_j(t_i;\beta)$ denotes the corresponding derivative of the state variables at $t_i$ given a fixed parameter $\beta$, i.e., there exists $\tilde{z}_j(t_i;\beta) \in \mathbb{C}$ such that \begin{align*}
&(\tu(t_i),\ty(t_i; \beta),\ldots, {\tu}^{(N)}(t_i), \ty^{(N)}(t_i; \beta), \tilde{z}_{N(L+1)}(t_i;\beta),\ldots,\tilde{z}_{j+1}(t_i;\beta), \tilde{z}_j(t_i;\beta) )\\
\in& \mathrm{V}(I^{(L)}\cap \mathbb{C}(a)[E_{j-1}]).%(I_{\beta,j-1}).
\end{align*}
  According to Proposition \ref{prop:elim}, $I^{(L)}\cap \mathbb{C}(a)[E_{j}] = \langle G_1, \ldots, G_s \rangle \cap \mathbb{C}(a)[E_{j}]$, which leads to $W_j = \mathrm{V}(I^{(L)} \cap \mathbb{C}(a)[E_{j}]) \cap \mathrm{V}(P_j)$. We denote $I_j$ as $I^{(L)} \cap \mathbb{C}(a)[E_{j}]$, the $j$-th elimination ideal of $I^{(L)}$.
  As $W_j = \phi$ implies that no partial solution
  \begin{align*}
    (\tilde{u}, \ldots, \tilde{u}^{(L-1)}, \tilde{y}, \ldots, \tilde{y}^{(L)}, \tilde{z}_{N(L+1)},\ldots,\tilde{z}_{j+1}) \in \mathrm{V}(I_j)
  \end{align*}
exists in $\mathrm{V}(P_j)$, given the arbitrary parameters of \eqref{eq:math_eq}, \eqref{eq:out_eq}, \eqref{eq:suffcond} becomes the sufficient condition for multivariate extension:
  \begin{align*}
    &\forall (\tilde{u}, \ldots, \tilde{u}^{(L-1)}, \tilde{y}, \ldots, \tilde{y}^{(L)}, \tilde{z}_{N(L+1)},\ldots,\tilde{z}_{j+1}) \in \mathrm{V}(I_{N(L+1)}),
    \exists (\tilde{z}_{N(L+1)},\ldots, \tilde{z}_{1}) \in \mathbb{C}^{N(L+1)}\\ & \mathrm{\, such\, that\, } (\tilde{u}, \ldots, \tilde{u}^{(L-1)}, \tilde{y}, \ldots, \tilde{y}^{(L)}, \tilde{z}_{N(L+1)},\ldots,\tilde{z}_{j+1}) \in \mathrm{V}(I^{(L)}).
    \end{align*}
If partial solution \eqref{eq:partial_io} is selected and can be extended to a complete solution of $\Il$, given $\beta$, $\tilde{z}_j's$ corresponding to the state variable derivatives must be those of \eqref{eq:diff_model} at $t_i$ given $\beta$ because \eqref{eq:math_eq}, \eqref{eq:out_eq} is assumed to admit a unique solution given a fixed parameter. This proves \eqref{eq:toshow} and hence, the theorem.
\end{proof}

To demonstrate the removal of the parameters not included in the parameter variety but in the algebraic variety computed through Algorithm \ref{algorithm:alg2}, Theorem \ref{thm:extention} is applied to \eqref{eq:virous} as follows.
\begin{example}\label{ex:singular}
As a result of Algorithm \ref{algorithm:alg2}, we get $l = 2$. The reduced Gr\"{o}bner basis of
\begin{align*}
    I^{(2)} := %\langle \rangle \\
    \langle
&\frac{\d x_2}{\d t} - \frac{a_4a_7}{a_6}x_3+a_4x_2, \frac{\d x_3}{\d t}-(1-a_5)a_6x_2+a_7x_3, &\nonumber \\
&\frac{\d^2 x_2}{\d t^2} - \frac{a_4a_7}{a_6}\frac{\d x_3}{\d t}+a_4\frac{\d x_2}{\d t}, \frac{\d^2 x_3}{\d t^2}-(1-a_5)a_6\frac{\d x_2}{\d t}+a_7\frac{\d x_3}{\d t},\\
&y-x_3, \frac{\d y}{\d t}-\frac{\d x_3}{\d t}, \frac{\d^2 y}{\d t^2}-\frac{\d^2 x_3}{\d t^2}\rangle &
\end{align*}
with respect to lexicographic ordering $\ddot{x}_3>\ddot{x}_2>\dot{x}_3>\dot{x}_2>x_3>x_2> \ddot{y}>\dot{y}>y$ is as follows:
\begin{align*}
G =& \left\{
\frac{\d^2y}{\d t^2}+(a_4+a_7)\frac{\d y}{\d t}+a_4a_5a_7y,
\quad(a_5a_6-a_6)x_2+\frac{\d y}{\d t}+a_7y,
\quad x_3-y,\right.\\
%&\left.a_6\frac{\d x_2}{\d t}+a_1a_3x_2-a_1a_4y,
&\left. (a_5a_6-a_6)\frac{\d {x}_2}{\d t}+(-a_4)\frac{\d^2 y}{\d t^2}+(-a_4a_5a_7)\frac{\d y}{\d t}, \right.\\
&\left. \quad \frac{\d x_3}{\d t}-\frac{\d y}{\d t},
\quad %a_3\frac{\d^2 x_2}{\d t^2}-a_1a_4\frac{\d x_3}{\d t}+a_1a_3\frac{\d x_2}{\d t}
(a_5a_6-a_6)\frac{\d^2 {x}_2}{\d t^2}+(a_4^2-a_4a_5a_7+a_4a_7)\frac{\d^2 y}{\d t^2}+(a_4^2a_5a_7)\frac{\d y}{\d t}
\right.\\
&\left.%\frac{\d^2 x_3}{\d t^2}-\frac{\d^2y}{\d t^2}
\frac{\d^2 {x}_3}{\d t^2}+(a_4+a_7)\frac{\d^2 y}{\d t^2}+(a_4a_5a_7)\frac{\d y}{\d t}
\right\}.
\end{align*}
Now, we introduce the following variables
\begin{align*}
    z_6 := x_2, z_5 := x_3, z_4 := \dot{x}_2, z_3:= \dot{x}_3, z_2:= \ddot{x}_2, z_1 := \ddot{x}_3.
\end{align*}
Let $I_{j}$ be the $j$-th elimination ideal of $I^{(2)}$, where $j = 1,\ldots, 2\times(2+1)$:
\begin{align*}
    I_{j} &= I^{(2)} \cap \mathbb{C}(a)[z_{j+1},\ldots,z_1, y, \ldots, y^{(L)}, u, \ldots, u^{(L-1)}] \\
    &= \langle G\rangle \cap \mathbb{C}(a)[z_{j+1},\ldots,z_1, y, \ldots, y^{(L)}, u, \ldots, u^{(L-1)}],
\end{align*}
for example,
\begin{align*}
    J = I_6 =& I^{(2)} \cap \mathbb{C}(a)[\ddot{y}, \dot{y}, y] \\%= I^{(2)} \cap \mathbb{C}(a)[z_0, \dot{y}, y]\\
    =& \langle \frac{\d^2y}{\d t^2}+(a_4+a_7)\frac{\d y}{\d t}+a_4a_5a_7y \rangle, \\
    I_{5} =& I_2 \cap \mathbb{C}(a)[x_2, \ddot{y}, \dot{y}, y] = I_2 \cap \mathbb{C}(a)[z_6, \ddot{y}, \dot{y}, y]\\
    =& \langle \frac{\d^2y}{\d t^2}+(a_4+a_7)\frac{\d y}{\d t}+a_4a_5a_7y, \textcolor{black}{(a_5a_6-a_6)z_6}+\frac{\d y}{\d t}+a_4y \rangle, \ldots.
\end{align*}
Thus, we get
\begin{align*}
    P_6 = \left\{ a_5a_6-a_6\right\}, P_5 = \left\{1\right\}, P_4 = \left\{a_5a_6 - a_6\right\}, P_3 = \left\{1\right\}, P_2 = \left\{a_5a_6 - a_6\right\}, P_1 = \left\{ 1\right\}
\end{align*}
assuming that $a_5a_6 -a_6 \neq 0$, i.e., $a_6 \neq 0, a_5 \neq 1$. %and an empty set $P_6$
As a result, for $j = 1,\ldots, 2\times(2+1)$
\begin{align*}
    W_j = \phi
\end{align*}
holds because $V(P_j) = \phi$. This shows that there is no parameter in the algebraic variety \eqref{eq:pv_virus} that is not in the parameter variety of \eqref{eq:math_eq}, \eqref{eq:out_eq}, given the data. Therefore, \eqref{eq:pv_virus} is coincident with the parameter variety in this case.
\end{example}

\section{Conclusion}\label{sec:conc}
In this study, we proposed a method for extracting all the feasible parameters that are uniquely determined by the observed time-series data and unidentifiable state-space models, explicitly and exhaustively. This parameter set was newly defined as the parameter variety of a model, given the time-series data. Our approach for the extraction of the parameter varieties is motivated by the challenges involved in dealing with unidentifiable models: the non-uniqueness of the feasible parameters causes the system properties, based on the estimated parameters, to depend on the estimated parameters.

The applicability of the proposed method was clearly demonstrated through viral dynamics analysis, given real data. In the application of our method to a model that represents the viral dynamics under therapy, the considerations on the viral dynamics that were overlooked were mentioned. Besides, by extracting the parameter varieties, the possible ranges of the parameters could be investigated. Furthermore, the possibility of geometric analysis of the parameter variety for capturing the abnormalities of the dynamics under consideration was also suggested.

In addition, we have theoretically justified the proposed method in which the input-output relationships of a model are investigated by examining the truncated differential ideals initially. Further, the observed time-series data or pseudo-data generated using the model are incorporated into the basis of the input-output relationships of the model; thereby, an algebraic variety containing the parameter variety is obtained. The truncation is determined by computing the minimal number of differentiations required to extract the input-output relationships from the state-space model. This truncation reasonably retains the assumptions on the smoothness of the model. Moreover, the algebraic variety containing the parameter variety of the model, given the data, is investigated under the framework of commutative algebra. This implies that computational algebraic techniques or software can be utilized to compute the algebraic variety containing the parameter variety and the parameters that are included in the algebraic variety but not in the parameter variety. This is important for practical application because our method can be applied without requiring detailed knowledge on algebra and algebraic geometry.

Although the incorporation of data into the basis of the input-output relationships of the model has been discussed in this study, it needs to be considered further. For example, it may be better to increase the number of measurement time points instead of using the minimal required, considering the robustness of the method against measurement noise. Besides, regarding pseudo-data generation using the models, the derivatives need be carefully estimated. Therefore, depending on the situation, it may be better to apply data-driven methods such as the Gaussian process for generating pseudo-data.
%Even though the applicability of the proposed method is confirmed through the example of viral dynamics analysis given real data, these future works makes the method real world applications with noisy data.

\section*{Acknowledgments}
This work was supported by JST CREST JPMJCR1914 and JSPS KAKENHI Grant Numbers JP20J21185 and 20K11693.

\appendix
\section{Exact solution to \eqref{eq:virous}}\label{sec:exact}
In this appendix, we show the exact solution for \eqref{eq:virous} \begin{align*}
    \frac{\d}{\d t}\begin{pmatrix}
    x_2\\
    x_3
    \end{pmatrix} = A
    \begin{pmatrix}
    x_2\\
    x_3
    \end{pmatrix},\quad A:= \begin{pmatrix}
    -a_4 & \frac{a_4a_7}{a_6}\\
    (1-a_5)a_6 & -a_7
    \end{pmatrix}
\end{align*}
under initial condition \eqref{eq:initcond}.
It can be confirmed in a straightforward manner that coefficient matrix $A$ has the following two eigenvalues %$\lambda_1, \lambda_2$
%Eigen values of $A$, we denote as, are as follows:
\begin{align*}
    \lambda_1 &:= -\left(\frac{a_{4}+a_7}{2}\right)-\frac{\sqrt{2\,a_{4}\,a_{7}+{a_{4}}^2+{a_{7}}^2-4\,a_{4}\,a_{5}\,a_{7}}}{2}, \\
   \lambda_2 &:= -\left(\frac{a_{4}+a_7}{2}\right) + \frac{\sqrt{2\,a_{4}\,a_{7}+{a_{4}}^2+{a_{7}}^2-4\,a_{4}\,a_{5}\,a_{7}}}{2},
\end{align*}
and the corresponding eigenvectors %corresponding to $\lambda_1, \lambda_2$ are
\begin{align*}
& {\left(\cfrac{\cfrac{a_{4}}{2}-\cfrac{a_{7}}{2}+\cfrac{\sqrt{2\,a_{4}\,a_{7}+{a_{4}}^2+{a_{7}}^2-4\,a_{4}\,a_{5}\,a_{7}}}{2}}{a_{6}\,\left(a_{5}-1\right)}, 1 \right)}^\top, \\
& {\left( \cfrac{\cfrac{a_{4}}{2}-\cfrac{a_{7}}{2}-\cfrac{\sqrt{2\,a_{4}\,a_{7}+{a_{4}}^2+{a_{7}}^2-4\,a_{4}\,a_{5}\,a_{7}}}{2}}{a_{6}\,\left(a_{5}-1\right)}, 1\right)}^\top.
\end{align*}
Hence, the solution $x_3(t)$ is given as
\begin{align}
\begin{aligned}
x_3(t) = &x_{3}(T_0)\left({\mathrm{e}}^{\left(t-T_{0}\right)\,\lambda_1}\,
\left(1-B\right) + {\mathrm{e}}^{\left(t-T_{0}\right)\,\lambda_2}\,B\right), \\
B &= \frac{%\left(
a_{4}+a_{7}-2\,a_{5}\,a_{7}+\sqrt{2\,a_{4}\,a_{7}+{a_{4}}^2+{a_{7}}^2-4\,a_{4}\,a_{5}\,a_{7}}%\right)
}{2\,\sqrt{2\,a_{4}\,a_{7}+{a_{4}}^2+{a_{7}}^2-4\,a_{4}\,a_{5}\,a_{7}}}. \label{eq:exact}
\end{aligned}
\end{align}

\end{document}